\documentclass[11pt]{article}

\newif\ifcomments
\commentstrue

\usepackage{fullpage}
\usepackage[utf8]{inputenc}

\usepackage{graphicx}
\usepackage{amsmath,amsthm,amssymb}
\usepackage{algorithm}
\usepackage{algpseudocode}
\usepackage{multirow}
\usepackage{makecell}

\usepackage{thm-restate,color,xspace}
\usepackage{comment}
\usepackage{thmtools}
\usepackage{xcolor}
\usepackage{nameref}
\usepackage{array}

\usepackage{tikz}
\usetikzlibrary{positioning}

\definecolor{ForestGreen}{rgb}{0.1333,0.5451,0.1333}
\definecolor{DarkRed}{rgb}{0.65,0,0}
\definecolor{Red}{rgb}{1,0,0}
\usepackage[linktocpage=true,
pagebackref=true,colorlinks,
linkcolor=DarkRed,citecolor=ForestGreen,
bookmarks,bookmarksopen,bookmarksnumbered]
{hyperref}
\usepackage{cleveref}
\usepackage{dsfont}

\declaretheorem[numberwithin=section]{theorem}
\declaretheorem[numberlike=theorem]{lemma}

\declaretheorem[numberlike=theorem]{proposition}

\declaretheorem[numberlike=theorem]{corollary}

\declaretheorem[numberlike=theorem]{conjecture}

\declaretheorem[numberlike=theorem]{claim}

\declaretheorem[numberlike=theorem,style=definition]{definition}
\declaretheorem[numberlike=theorem,style=definition]{remark}

\global\long\def\polylog{\mathrm{polylog}}
\global\long\def\polyloglog{\mathrm{polyloglog}}
\global\long\def\ET{\mathrm{ET}}

\global\long\def\ssum{{\rm sum}}
\global\long\def\on{{\rm on}}
\global\long\def\off{{\rm off}}
\global\long\def\Table{\mathsf{Table}}

\global\long\def\aff{{\rm aff}}
\global\long\def\new{{\rm new}}
\global\long\def\Boruvka{Bor{\r u}vka}
\global\long\def\act{\rm act}
\global\long\def\CountAll{\mathsf{CountAll}}
\global\long\def\CountA{\mathsf{CountA}}
\global\long\def\CountB{\mathsf{CountB}}
\global\long\def\init{\rm init}

\newcommand{\wtilde}{\widetilde}
\newcommand{\what}{\widehat}

\newcommand{\Otil}{\widetilde{O}}
\newcommand{\Ohat}{\widehat{O}}
\newcommand{\Omegahat}{\hat{\Omega}}
\global\long\def\poly{\mathrm{poly}}

\global\long\def\A{\mathcal{A}}

\global\long\def\Table{\mathsf{Table}}

\global\long\def\val{\mathrm{val}}
\global\long\def\it{\mathrm{int}}
\global\long\def\drop{\mathrm{drop}}
\global\long\def\SFDecomp{\textsc{SFDecomp}}
\global\long\def\bool{\rm bool}

\newcommand{\ignore}[1]{}

\title{Better Decremental and Fully Dynamic Sensitivity Oracles for Subgraph Connectivity}
\author{Yaowei Long\thanks{yaoweil@umich.edu}\\ University of Michigan\and Yunfan Wang\thanks{yunfan-w20@mails.tsinghua.edu.cn}\\ Tsinghua University}
\begin{document}

\maketitle

\begin{abstract}
We study the \emph{sensitivity oracles problem for subgraph connectivity} in the \emph{decremental} and \emph{fully dynamic} settings. In the fully dynamic setting, we preprocess an $n$-vertices $m$-edges undirected graph $G$ with $n_{\off}$ deactivated vertices initially and the others are activated. Then we receive a single update $D\subseteq V(G)$ of size $|D| = d \leq d_{\star}$, representing vertices whose states will be switched. Finally, we get a sequence of queries, each of which asks the connectivity of two given vertices $u$ and $v$ in the activated subgraph. The decremental setting is a special case when there is no deactivated vertex initially, and it is also known as the \emph{vertex-failure connectivity oracles} problem.

We present a better deterministic vertex-failure connectivity oracle with $\what{O}(d_{\star}m)$ preprocessing time, $\wtilde{O}(m)$ space, $\wtilde{O}(d^{2})$ update time and $O(d)$ query time, which improves the update time of the previous almost-optimal oracle \cite{long2022near} from $\what{O}(d^{2})$ to $\wtilde{O}(d^{2})$. 

We also present a better deterministic fully dynamic sensitivity oracle for subgraph connectivity with $\what{O}(\min\{m(n_{\off} + d_{\star}),n^{\omega}\})$ preprocessing time, $\wtilde{O}(\min\{m(n_{\off} + d_{\star}),n^{2}\})$ space, $\wtilde{O}(d^{2})$ update time and $O(d)$ query time, which significantly improves the update time of the state of the art \cite{HKP23} from $\wtilde{O}(d^{4})$ to $\wtilde{O}(d^{2})$. Furthermore, our solution is even almost-optimal assuming popular fine-grained complexity conjectures.

\end{abstract}

\clearpage

\section{Introduction}

We study the \emph{sensitivity oracles problem for subgraph connectivity} in the \emph{decremental} and \emph{fully dynamic} settings, which is one of the fundamental dynamic graph problems in undirected graphs. In the fully dynamic setting, this problem has three phases. In the preprocessing phase, given an integer $d_{\star}$, we preprocess an $n$-vertices $m$-edges undirected graph $G=(V,E)$ in which some vertices are \emph{activated}, called \emph{on-vertices} and denoted by $V_{\on}$, while the others are \emph{deactivated}, called \emph{off-vertices} and denoted by $V_{\off}$. We let $n_{\on} = |V_{\on}|$ and $n_{\off} = |V_{\off}|$ denote the number of initial on-vertices and off-vertices respectively. In the update phase, we will receive a set $D\subseteq V$ with $|D| = d\leq d_{\star}$, representing vertices whose states will be switched, and we update the oracle. In the subsequent query phase, let 
$V_{\new} = (V_{\on}\setminus D)\cup(V_{\off}\cap D)$ denote the activated vertices after the update. Each query will give a pair of vertices $u,v\in V_{\new}$ and ask the connectivity of $u$ and $v$ in the new activated subgraph $G[V_{\new}]$. The decremental setting is a special case where there is no off-vertices initially, i.e. $V_{\off}$ is empty.

The decremental version of this problem is also called the \emph{vertex-failure connectivity oracles} problem, which has been studied extensively, e.g. \cite{duan2010connectivity,duan2020connectivity,BrandS19,PilipczukSSTV22,long2022near,Kos23}, and its complexity was well-understood up to subpolynomial factors. Specifically, a line of works by Duan and Pettie \cite{duan2010connectivity,duan2020connectivity} started the study on vertex-failure connectivity oracles for general $d_{\star}$, and they finally gave a deterministic oracle with $\wtilde{O}(mn)$ preprocessing time, $\wtilde{O}(md_{\star})$ space, $\wtilde{O}(d^{2})$ update time and $O(d)$ query time. Following \cite{duan2010connectivity,duan2020connectivity}, Long and Saranurak \cite{long2022near} presented an improved solution with $\what{O}(m) + \wtilde{O}(md_{\star})$ preprocessing time, $\wtilde{O}(m)$ space, $\what{O}(d^{2})$ update time and $O(d)$ query time\footnote{Throughout the paper, we use $\Otil(\cdot)$ to hide a $\polylog(n)$
factor and use $\Ohat(\cdot)$ to hide a $n^{o(1)}$ factor.}, which is optimal up to subpolynomial factors because matching (conditional) lower bounds for all the four complexity measurements were shown in \cite{HenzingerKNS15,duan2020connectivity,long2022near}. We refer to \Cref{tab:compare} for more solutions to this problem (for example, Kosinas \cite{Kos23} proposed a simple and practical algorithm using a conceptually different approach). However, there are still relatively large subpolynomial overheads on the current almost-optimal upper bounds (i.e., on the preprocessing time and update time of the LS-oracle), so a natural question is whether we can improve them:
\begin{center}
\emph{Can we design an almost-optimal deterministic vertex-failure connectivity oracle with only\\ polylogarithmic overheads on the update time, or even on all complexity bounds?}
\end{center}

The fully dynamic sensitivity oracles problem for subgraph connectivity was studied by \cite{henzinger2016incremental,HKP23}. Henzinger and Neumann \cite{henzinger2016incremental} showed a black-box reduction from the decremental setting. Plugging in the almost-optimal decremental algorithm of \cite{long2022near}, this reduction leads to a fully dynamic sensitivity oracle with $\what{O}(n_{\off}^{2}m) + \wtilde{O}(d_{\star}n_{\off}^{2}m)$ preprocessing time, $\wtilde{O}(n_{\off}^{2}m)$ space, $\what{O}(d^{4})$ update time and $O(d^{2})$. Hu, Kosinas and Polak \cite{HKP23} studied this problem from an equivalent but different perspective called \emph{connectivity oracles for predictable vertex failures}, which gave a solution with $\wtilde{O}((n_{\off} + d_{\star})m)$ preprocessing time, $\wtilde{O}((n_{\off} + d_{\star})m)$ space, $\wtilde{O}(d^{4})$ update time and $O(d)$ query time\footnote{In \cite{HKP23}, they modeled the problem using slightly different parameters, and here we describe their bounds using our parameters. Basically, they defined a parameter $d'$ (named $d$ in their paper) in the preprocessing phase and $\eta$ in the update phase. Their $\eta$ is equivalent to our $d$. Besides, fixing $d_{\star}$ and $n_{\off}$, our input instances can be reduced to theirs with $d'$ at most $n_{\off} + d_{\star}$. In the other direction, fixing $d'$, their input instances can be reduced to ours with $d_{\star} + n_{\off}$ at most $3d'$. Furthermore, their space complexity was not specified, but to our best knowledge, it should be roughly proportional to their preprocessing time.}. Despite the efforts, there are still gaps between the upper bounds of the fully dynamic setting and the decremental setting (except the query time). Naturally, one may have the following question:
\begin{center}
\emph{Can we match the upper bounds in the fully dynamic and decremental settings\\ or show separations between them for all the four measurements?}
\end{center}
Notably, \cite{HKP23} showed a conditional lower bound $\what{\Omega}((n_{\off} + d_{\star})m)$ on the preprocessing time\footnote{This lower bound is obtained from input instances with $n_{\off} = \Theta(n)$ and $m$ is roughly linear to $n$. See the discussion in \Cref{remark:} in \Cref{sect:LB}.}, which separated two settings at the preprocessing time aspect. However, the right complexity bounds are still not clear for the space and the update time. In particular, given that two different approaches \cite{henzinger2016incremental,HKP23} both showed upper bounds around $d^{4}$ for update time, it is interesting to identify if this is indeed a barrier or it just happened accidentally. Furthermore, we note that it seems hard to improve the update time following either of these two approaches, because the black-box reduction by \cite{henzinger2016incremental} has been plugged in the almost-optimal decremental oracle, and the fully dynamic oracle by \cite{HKP23} generalizes the decremental oracle by \cite{Kos23}, where the latter already has $\wtilde{O}(d^{4})$ update time.

\subsection{Our Results}

We give a partially affirmative answer to the first question and answer the second question affirmatively by the following results.

\paragraph{The Decremental Setting.} We show a better vertex-failure connectivity oracle by improving the $\what{O}(d^{2})$ update time of the current almost-optimal solution \cite{long2022near} to $\wtilde{O}(d^{2})$. See \Cref{coro:VFConnOracle} for a detailed version of \Cref{thm:DecrementalMain}.

\begin{theorem}
There exists a deterministic vertex-failure connectivity oracle with $\what{O}(m) + \wtilde{O}(d_{\star}m)$ preprocessing time, $\wtilde{O}(m)$ space, $\wtilde{O}(d^{2})$ update time and $O(d)$ query time.
\label{thm:DecrementalMain}
\end{theorem}

Same as all the previous vertex-failure connectivity oracles, we can substitute all the $m$ factors in \Cref{thm:DecrementalMain} with $\bar{m} = \min\{m,n(d_{\star}+1)\}$ using a standard sparsification by Nagamochi and Ibaraki \cite{NI92} at a cost of an additional $O(m)$ preprocessing time.

\begin{table}[h]
\footnotesize{

\begin{tabular}{|>{\raggedright}p{0.18\textwidth}|c|c|c|c|c|}
\hline 
 & \makecell{Det./\\Rand.} & Space & Preprocessing & Update & Query\tabularnewline
\hline 
\hline 
Block trees, \\
SQRT trees, and \cite{kanevsky1991line} \\
only when $d_{\star}\le3$ & Det. & $O(n)$ & $\Otil(m)$ & $O(1)$ & $O(1)$\tabularnewline
\hline 
\makecell[l]{Duan \& Pettie\\ \cite{duan2010connectivity} for $c\ge1$} & Det. & \makecell{linear in\\ preprocessing time} & $\Otil(m d_{\star}^{1-\frac{2}{c}}n^{\frac{1}{c}-\frac{1}{c\log(2d_{\star})}})$ & $\Otil(d^{2c+4})$ & $O(d)$\tabularnewline
\hline 
\multirow{2}{0.18\textwidth}{Duan \& Pettie \cite{duan2020connectivity}} & Det. & $O(m d_{\star}\log n)$ & $O(m n\log n)$ & $O(d^{3}\log^{3}n)$ & $O(d)$\tabularnewline
\cline{2-6} 
 & Rand. & $O(m\log^{6}n)$ & $O(m n\log n)$ & $\bar{O}(d^{2}\log^{3}n)$ w.h.p. & $O(d)$\tabularnewline
\hline 
Brand \& Saranurak \cite{BrandS19}
& Rand. & $O(n^{2})$ & $O(n^{\omega})$ & $O(d^{\omega})$ & $O(d^{2})$\tabularnewline
\hline 
\multirow{2}{0.18\textwidth}{Pilipczuk et al. \cite{PilipczukSSTV22}} & Det. & $m2^{2^{O(d_{\star})}}$ & $m n^{2}2^{2^{O(d_{\star})}}$ & $2^{2^{O(d_{\star})}}$ & $2^{2^{O(d_{\star})}}$\tabularnewline
\cline{2-6} 
 & Det. & $n^{2}\poly(d_{\star})$ & $\poly(n)2^{O(d_{\star}\log d_{\star})}$ & $\poly(d_{\star})$ & $\poly(d_{\star})$\tabularnewline
\hline 
\multirow{2}{0.18\textwidth}{Long \& Saranurak \cite{long2022near}} 
& Det. & $O(m\log^{3}n)$ & $O(m n\log n)$ & $\bar{O}(d^{2}\log^{3}n\log^{4}d)$ & $O(d)$\tabularnewline
\cline{2-6}
& Det. & $O(m\log^{*}n)$ & $\Ohat(m) + \Otil(d_{\star}m)$ & $\Ohat(d^{2})$ & $O(d)$\tabularnewline
\hline 
Kosinas \cite{Kos23} & Det. & $O(d_{\star}m\log n)$ & $O(d_{\star}m\log n)$ & $O(d^{4}\log n)$ & $O(d)$ \tabularnewline
\hline
\textbf{This paper} & Det. & $O(m\log^{3} n)$ & $\what{O}(m) + O(d_{\star}m\log^{3} n)$ & $O(d^{2}(\log^{7}n + \log^{5}n\log^{4}d))$ & $O(d)$ \tabularnewline
\hline
\end{tabular}

}

\caption{Complexity of known vertex-failure connectivity oracles. All the $m$ factors can be replaced by $\bar{m}=\min\{m,n(d_{\star}+1)\}$ at a cost of an additional $O(m)$ preprocessing time. The randomized algorithms are all Monte Carlo.  The notation $\bar{O}(\cdot)$ hides a $\polyloglog(n)$ factor.\label{tab:compare}
}
\end{table}

We emphasize that our result is a strict improvement on \cite{long2022near}. In addition to the improvement on update time, our algorithm also improves the hidden subpolynomial overheads on the preprocessing time. We achieve this by giving a new construction algorithm of the \emph{low degree hierarchy}, a graph decomposition technique widely used in this area \cite{duan2020connectivity,long2022near,PPP23}. Roughly speaking, the previous almost-linear-time construction \cite{long2022near} relies on modern graph techniques including vertex expander decomposition and approximate vertex capacitated maxflow, which are highly complicated and will bring relatively large subpolynomial overheads. Our new construction bypasses the vertex expander decomposition to obtain improvement on both efficiency and quality, which is also considerably simpler. Finally, we point out that the subpolynomial factors in our preprocessing time still comes from the construction of the low degree hierarchy, which can be traced back to the subpolynomial overheads of the current approximate vertex capacitated maxflow algoirthm \cite{BGS21}.

\paragraph{The Fully Dynamic Setting.} We also show a better fully dynamic sensitivity oracle for subgraph connectivity, with update time and query time matching the decremental bounds up to polylogarithmic factors. See \Cref{thm:FullyDynamicOracle} for a detailed version of \Cref{thm:FullyDynamicMain}. The first upper bound $\what{O}(m) + \wtilde{O}(m(n_{\off} + d_{\star}))$ is obtained from a combinatorial algorithm\footnote{\emph{Combinatorial} algorithms \cite{AW14} are algorithms
that do not use fast matrix multiplication.}.

\begin{theorem}
There exists a deterministic fully dynamic sensitivity oracle for subgraph connectivity with $\what{O}(m) + \wtilde{O}(\min\{m(n_{\off} + d_{\star}),n^{\omega}\})$ preprocessing time, $\wtilde{O}(\min\{m(n_{\off} + d_{\star}),n^{2}\})$ space, $\wtilde{O}(d^{2})$ update time and $O(d)$ query time, where $\omega$ is the exponent of matrix multiplication.
\label{thm:FullyDynamicMain}
\end{theorem}

\begin{table}[h]
\centering
\footnotesize{

\begin{tabular}{|>{\raggedright}p{0.15\textwidth}|c|c|c|c|c|}
\hline 
 & \makecell{Det./\\Rand.} & Space & Preprocessing & Update & Query\tabularnewline
\hline 
\hline 
\makecell[l]{Henzinger \&\\ Neumann \cite{henzinger2016incremental}} & Det. & $\wtilde{O}(n_{\off}^{2}m)$ & $\what{O}(n_{\off}^{2}m) + \wtilde{O}(d_{\star}n_{\off}^{2}m)$ & $\what{O}(d^{4})$ & $O(d^{2})$ \tabularnewline
\hline
\makecell[l]{Hu, Kosinas \&\\ Polak \cite{HKP23}} & Det. & $\wtilde{O}((n_{\off} + d_{\star})m)$ & $\wtilde{O}((n_{\off} + d_{\star})m)$ & $\wtilde{O}(d^{4})$ & $O(d)$ \tabularnewline
\hline
\textbf{This paper} & Det. & $O(\min\{(n_{\off} + d_{\star})m\log^{2}n,n^{2}\})$ & \makecell{$\what{O}(m) +$\\ $O(\min\{(n_{\off} + d_{\star})m, n^{\omega}\}\log^{2}n)$} & $O(d^{2}\log^{7}n)$ & $O(d)$ \tabularnewline
\hline
\end{tabular}

}

\caption{Complexity of known fully dynamic sensitivity oracles for subgraph connectivity.\label{tab:fullydynamic}
}
\end{table}

We also show conditional lower bounds on the preprocessing time and the space, which separate the fully dynamic and decremental settings. Furthermore, combining our new lower bounds and the existing ones, our solution in \Cref{thm:FullyDynamicMain} is optimal up to subpolynomial factors.

\begin{theorem}
Let ${\cal A}$ be a fully dynamic sensitivity oracle for subgraph connectivity with $S$ space, $t_{p}$ preprocessing time, $t_{u}$ update time and $t_{q}$ query time upper bounds. Assuming popular conjectures, we have the following:
\begin{enumerate}
\item If $t_{u} + t_{p} = f(d)\cdot n^{o(1)}$, then $S = \what{\Omega}(n^{2})$. (See \Cref{lemma:SpaceLB} and \Cref{coro:SpaceLBAllDensity})
\item If $t_{u} + t_{p} = f(d)\cdot n^{o(1)}$, then $t_{u} = \what{\Omega}((n_{\off} + d)m)$ (See \cite{HKP23})
\item If $t_{u} + t_{p} = f(d)\cdot n^{o(1)}$, then $t_{u} = \what{\Omega}(n^{\omega_{\bool}})$ (See \Cref{lemma:PreprocessLBBMM})
\item If $t_{p} = \poly(n)$, then $t_{u} + t_{q} = \what{\Omega}(d^{2})$. (See \cite{long2022near})
\item If $t_{p} = \poly(n)$ and $t_{u} = \poly(dn^{o(1)})$, then $t_{q} = \what{\Omega}(d)$. (See \cite{HenzingerKNS15})
\end{enumerate}
The $f(d)$ above can be an arbitrary growing function, and $\omega_{\bool}$ is the exponent of Boolean matrix multiplication.
\end{theorem}

We make some additional remarks here. When discussing lower bounds, we assume $d = d_{\star}$ for each update. The lower bound on the space (item 1) holds even when the input graphs are sparse, so it naturally holds for input graphs with general density, as shown in \Cref{coro:SpaceLBAllDensity}. The lower bounds on the preprocessing time (items 2 and 3) are not contradictory, because item 2 is obtained from sparse graphs while item 3 is obtained from dense graphs, as discussed in \Cref{remark:}. If we focus on combinatorial algorithms, we can even obtain a $\what{O}((n_{\off} + d)m)$ lower bound for general density, as discussed in \Cref{remark:CombUP}. The lower bounds on the update time and query time (items 4 and 5) are from those in the decremental setting.

\subsection{Organization}

In \Cref{sect:Overview}, we give an overview of our techniques. In \Cref{sect:Prelim}, we give the preliminaries. In \Cref{sect:LowDegreeHierarchy}, we introduce our new construction of the low degree hierarchy and obtain a better vertex-failure connectivity oracle as a corollary. In \Cref{sect:Preprocessing,sect:UpdateQuery}, we describe the preprocessing, update and query algorithms of our fully dynamic sensitivity oracle for subgraph connectivity. In \Cref{sect:LB}, we discuss lower bounds in the fully dynamic setting. We includes all omitted proofs in \Cref{sect:OmittedProofs}.

\section{Technical Overview}
\label{sect:Overview}

\paragraph{Better Vertex-Failure Connectivity Oracles.} Our main contribution is a new construction of the \emph{low degree hierarchy}. Then we obtain a better vertex-failure connectivity oracle as a corollary by combining the new construction of the low degree hierarchy and the remaining part in \cite{long2022near}.

It is known that the construction of the low degree hierarchy can be reduced to $O(\log n)$ calls to the \emph{low-degree Steiner forest decomposition} \cite{duan2020connectivity,long2022near}. Basically, for an input graph $G$ with terminal set $U\subseteq V(G)$, we say a forest $F\subseteq E(G)$ is a \emph{Steiner forest} of $U$ in $G$ if $F$ spans the whole $U$ (may also span some additional non-$U$ vertices) and for each $u,v\in U$, $u,v$ are connected in $F$ if and only if they are connected in $G$. 
We propose a new almost-linear time low-degree Steiner forest decomposition algorithm as shown in \Cref{lemma:SFDecompOverview}, which improves the degree parameter $\Delta$ from $n^{o(1)}$ to $O(\log^{2}n)$ compared to the previous one by \cite{long2022near}. This will leads to an improvement to the quality of the low degree hierarchy, and finally reflects on the update time. 

\begin{lemma}[\Cref{lemma:SFDecomp}, Informal]
Let $G$ be an undirected graph with terminals $U\subseteq V(G)$. There is an almost-linear-time algorithm that computes a \emph{separator} $|X|\subseteq V(G)$ of size $|X|\leq |U|/2$, and a low-degree Steiner forest of $U\setminus X$ in $G\setminus X$ with maximum degree $\Delta = O(\log^{2}n)$.

\label{lemma:SFDecompOverview}
\end{lemma}

In the following discussion, we assume $U = V(G)$ for simplicity (hence spanning trees/forests and Steiner trees/forests are now interchangable). To obtain \Cref{lemma:SFDecompOverview}, our starting point is that it is \emph{not} necessary to perform a vertex expander decomposition (which will bring large $n^{o(1)}$ overheads to $\Delta$) to get a low-degree Steiner forest decomposition. Basically, in \cite{long2022near}, they obtain a fast low-degree Steiner forest decomposition by first proving that any vertex expander admits a low-degree spanning tree, so then it suffices to perform the stronger vertex expander decomposition. The way they prove the former is to argue that for any vertex expander $H$, one can embed another expander $W$ into $H$ with low vertex-congestion, which implies that $H$ has a low-degree subgraph including all vertices in $V(H)$.

The key observation is that, to make the above argument work, $W$ is no need to be an expander and $W$ can be an arbitrary \emph{connected graph}. This inspires us to design the following subroutine \Cref{lemma:CutOrSteinerTreeOverview}. Then \Cref{lemma:SFDecompOverview} can be shown by invoking \Cref{lemma:CutOrSteinerTree} using a standard divide-and-conquer framework.

\begin{lemma}[\Cref{lemma:CutOrSteinerTree}, Informal]
Let $G$ be an undirected graph. There is an almost-linear-time algorithm that computes either
\begin{itemize}
\item a balanced sparse vertex cut $(L,S,R)$ with $|R|\geq |L|\geq |V(G)|/12$ and $|S|\leq 1/(100\log n)\cdot |L|$.
\item a large subset $V'\subseteq V(G)$ with $|V'|\geq 3|U|/4$ s.t. we can embed a connected graph $W'$ with $V(W') = V'$ into $G[V']$ with vertex congestion $O(\log^{2} n)$, which implies a spanning in $G[V']$ with maximum degree $O(\log^{2}n)$.
\end{itemize}
\label{lemma:CutOrSteinerTreeOverview}
\end{lemma}

We design the algorithm in \Cref{lemma:CutOrSteinerTreeOverview} using a simplified \emph{cut-matching game}. The original cut-matching game \cite{KRV09,KKOV07} can be used to embed an expander into a graph with low congestion (or produce a balanced sparse cut). To embed a connected graph, consider the following procedure. Assume a standard matching player (i.e. \Cref{lemma:MatchingPlayer}) which, given a graph $G$ and a balanced partition $(A,B)$ of $V(G)$, either embeds a large matching between $A$ and $B$ into $G$ with low vertex-congestion or outputs a balanced sparse vertex cut in almost-linear time. Start with a graph $W$ with $V(W) = V(G)$ but no edge and perform several rounds. At each round, we (as the cut player) partition the connected components of $W$ into two parts with balanced sizes, and feed the partition to the matching player. If the matching player gives a matching, we add it to $W$ and go to the next round. The game stops once a giant connected component (of size at least $3|V(G)|/4$) appears in $W$, which will roughly serve as $W'$. Roughly speaking, the game will stop in $O(\log n)$ rounds because at each round, there exists a large fraction of vertices, s.t. for each of them (say vertex $v$), the component containing $v$ has its size doubled.

\paragraph{Fully Dynamic Sensitivity Oracles for Subgraph Connectivity.}

Our fully dynamic oracle is actually a generalization of a simplified version of the decremental oracle in \cite{long2022near}. 

Initially, we construct a low degree hierarchy on the activated subgraph $G_{\on}:=G[V_{\on}]$. As mentioned in \cite{long2022near}, the hierarchy will roughly reduce $G_{\on}$ to the following \emph{semi-bipartite} form. First, $V_{\on}$ can be partitioned into $L_{\on}$ and $R_{\on}$, called \emph{left on-vertices} and \emph{right on-vertices} respectively, s.t. there is no edge connecting two vertices in $R_{\on}$. Second, $L_{\on}$ is spanned by a known path $\tau$. Therefore, we assume the original graph $G$ has a semi-bipartite $G_{\on}$ from now. 

When there is no off-vertices initially (i.e. the decremental setting), the properties of a semi-bipartite $G_{\on}$ naturally support the following update and query strategy. In the update phase, removing vertices in $D$ will break the path $\tau$ into at most $d+1$ \emph{intervals}, and we will somehow (we will not explain this in the overview) recompute the connectivity of these intervals in the graph $G_{\on}\setminus D$. Then, for each query of $u,v\in V_{\on}\setminus D$, it suffices to find two intervals $I_{u},I_{v}$ connecting with $u,v$ in $G_{\on}\setminus D$ respectively. When $u,v$ are left on-vertices, $I_{u},I_{v}$ can be found trivially. When $u,v$ are right on-vertices, we just need to scan at most $d+1$ neighbors of each of $u,v$, which takes $O(d)$ time. Note that removing $D$ will generate at most $d+1$ intervals is a crucial point to achieve fast update time.

Back to the fully dynamic setting, for an update $D$, in addition to removing vertices $D_{\on}:=D\cap V_{\on}$ from $G_{\on}$, we will also add vertices $D_{\off}:=D\cap V_{\off}$ and their incident edges. The key observation is that $G[V_{\on}\cup D_{\off}]$ is still roughly a semi-bipartite graph. The first property will still hold if we put the newly activated vertices $D_{\off}$ into the left side. The second property may not hold because we do not have a path spanning the new left vertices $L_{\on}\cup D_{\off}$. However, this will not hurt because we can still partition 
$L_{\on}\cup D_{\off}$ into $O(d)$ connected parts after removing $D_{\on}$ from $G[V_{\on}\cup D_{\off}]$, i.e. at most $d+1$ intervals covering $L_{\on}\setminus D_{\on}$, and at most $d$ vertices in $D_{\off}$. 

Giving this key observation, it is quite natural to adapt the decremental algorithm to the fully dynamic setting. Using the ideas of \emph{adding artificial edges} (intuitively, substituting each right vertex and its incident edges with an artificial clique on its left neighbors) and applying \emph{2D range counting structure}, we can design an update algorithm with $\wtilde{O}(d^{3})$ update time \cite{duan2020connectivity}. To improve the update time to $\wtilde{O}(d^{2})$, we can use a \emph{\Boruvka's styled} update algorithm and implement it by considering \emph{batched adjacency queries} on intervals \cite{long2022near}.

\section{Preliminaries}
\label{sect:Prelim}

Throughout the paper, we use the standard graph theoretic notation. For any graph, we use $V (\cdot)$
and $E(\cdot)$ to denote its vertex set and edge set respectively. 
If If there is no other specification, we use $G$ to denote the original graph on which we will build the oracle, and we let $n = |V (G)|$ and $m = |E(G)|$. Initially, the vertices $V(G)$ in the original graph are partitioned into \emph{on-vertices} $V_{\on}$ and \emph{off-vertices} $V_{\off}$, and we let $n_{\off} = |V_{\off}|$. For a graph $H$ and any $S \subseteq V(H)$, we let $H[S]$ denote the subgraph induced by vertices $S$. Also, for any $S \subseteq V(H)$, we use $H \setminus S$ to denote the graph after removing vertices in $S$ and edges incident to them. Similarly, for any $F \subseteq E(H)$, $G \setminus F$ denote the graph after removing edges in $F$. 

We also use the notion of \emph{multigraphs}. For a multigraph $H$, its edge set $E(H)$ is a \emph{multiset}. We use $+$ and $\sum$ to denote the union operation and use $-$ to denote the subtraction operation on multiset. We let $\omega$ denote the exponent of matrix multiplication and $\omega_{\bool}$ denote the exponent of Boolean matrix multiplication. To our best knowledge, currently $\omega$ and $\omega_{\bool}$ have the same upper bound.

\section{The Low Degree Hierarchy}
\label{sect:LowDegreeHierarchy}

The \emph{low degree hierarchy} was first introduced in \cite{duan2020connectivity} to design efficient vertex-failure connectivity oracles. The construction of this hierarchy in \cite{duan2020connectivity} is based on the approximate minimum degree Steiner forest algorithm of \cite{furer1994approximating}, which gives $\tilde{O}(mn)$ construction time. Later, an alternative construction algorithm was shown in \cite{long2022near} by exploiting vertex expander decomposition, which improves the construction time to $m^{1+o(1)}$, at a cost of a small quality loss.

In this section, we will show a new construction algorithm, which still runs in almost-linear time and gives a hierarchy with quality better than the one in \cite{long2022near} (but still worse than the one in \cite{duan2020connectivity}). To obtain the quality improvement, we basically simplify the construction in \cite{long2022near} and bypass the vertex expander decomposition. 

We define the low degree hierarchy in \Cref{def:LowDegreeHierarchy}, and the main result of this section is \Cref{thm:LowDegreeHierarchy}. It was known that constructing a low degree hierarchy reduces to several rounds of \emph{low-degree Steiner forest decomposition}. In \Cref{sect:BalCutorLDST}, we introduce our key subroutine \Cref{lemma:CutOrSteinerTree}, which given an input graph, either computes a balanced sparse vertex cut or a low-degree Steiner tree covering a large fraction of terminals. In \Cref{sect:SFDecomp}, we show the low-degree Steiner forest decomposition algorithm \Cref{lemma:SFDecomp} using \Cref{lemma:CutOrSteinerTree} in a standard divide and conquer framework, and then complete the proof of \Cref{thm:LowDegreeHierarchy}.

\begin{definition}[Low Degree Hierarchy \cite{duan2020connectivity}, Definition 5.1 in \cite{long2022near}]
Let $G$ be a connected undirected graph. A $(p,\Delta)$-low degree hierarchy with height $p$ and degree parameter $\Delta$ on $G$ is a pair $({\cal C},{\cal T})$ of sets, where ${\cal C}$ is a set of vertex-induced connected subgraphs called \textit{components}, and ${\cal T}$ is a set of Steiner trees with maximum vertex degree at most $\Delta$.

The set ${\cal C}$ of components is a laminar set. Concretely, it satisfies the following properties.

\begin{itemize}
    \item[(1)] Components in ${\cal C}$ belong to $p$ levels and we denote by ${\cal C}_{i}$ the set of components at level $i$. In particular, at the top level $p$, ${\cal C}_{p}=\{G\}$ is a singleton set with the whole $G$ as the unique component. Furthermore, for each level $i\in[1,p]$, components in ${\cal C}_{i}$ are vertex-disjoint and there is no edge in $E(G)$ connecting two components in ${\cal C}_{i}$.
    \item[(2)] For each level $i\in [1,p-1]$ and each component $\gamma\in{\cal C}_{i}$, there is a unique component $\gamma'\in{\cal C}_{i+1}$ such that $V(\gamma)\subseteq V(\gamma')$, where we say that $\gamma'$ is the \textit{parent-component} of $\gamma$ and that $\gamma$ is a \textit{child-component} of $\gamma'$.
    \item[(3)] For each component $\gamma\in{\cal C}$, the \textit{terminals of $\gamma$}, denoted by $U(\gamma)$, are vertices in $\gamma$ but not in any of $\gamma$'s child-components. Note that $U(\gamma)$ can be empty. In particular, for each $\gamma\in{\cal C}_{1}$, $U(\gamma)=V(\gamma)$.
\end{itemize}

Generally, for each level $i\in[1,p]$, we define the \textit{terminals at level $i$} be terminals in all components in ${\cal C}_{i}$, denoted by $U_{i}=\bigcup_{\gamma\in{\cal C}_{i}}U(\gamma)$.

The set ${\cal T}$ of low-degree Steiner trees has the following properties.
\begin{itemize}
    \item[(4)] ${\cal T}$ can also be partitioned into subsets ${\cal T}_{1},...,{\cal T}_{p}$, where ${\cal T}_{i}$ denote trees at level $i$ and trees in ${\cal T}_{i}$ are vertex-disjoint.
    \item[(5)] For each level $i\in[1,p]$ and tree $\tau\in{\cal T}_{i}$, the \textit{terminals of $\tau$} is defined by $U(\tau)=U_{i}\cap V(\tau)$.
    \item[(6)] For each level $i\in[1,p]$ and each component $\gamma\in{\cal C}_{i}$ with $U(\gamma)\neq\emptyset$, there is a tree $\tau\in{\cal T}_{i}$ such that $U(\gamma)\subseteq U(\tau)$, denoted by $\tau(\gamma)$. We emphasize that two different components $\gamma$ and $\gamma' \in {\cal C}_i$ may correspond to the same tree $\tau \in {\cal T}_i$. 
\end{itemize}
\label{def:LowDegreeHierarchy}
\end{definition}

For better understanding, we note that the terminal sets of component $\{U(\gamma)\mid \gamma\in{\cal C}\}$, levels $\{U_{i}\mid 1\leq i\leq p\}$, and Steiner trees $\{U(\tau)\mid \tau\in{\cal T}\}$ are all partitions of $V(G)$. One may also get the picture of the hierarchy from the perspective of construction. See the construction described in \Cref{algo:LowDegreeHierarchy}, which invokes \Cref{lemma:SFDecomp} in a black-box way.

\begin{theorem}
Let $G$ be an undirected graph. There is a deterministic algorithm that computes a $(p,\Delta)$-low degree hierarchy with $p = O(\log n)$ and $\Delta = O(\log^{2} n)$. The running time is $m^{1+o(1)}$.
\label{thm:LowDegreeHierarchy}
\end{theorem}

\Cref{coro:VFConnOracle} is obtained by substituting the construction of low degree hierarchy in \cite{long2022near} with ours. Formally speaking, it is a corollary of \Cref{thm:LowDegreeHierarchy}, and Lemma 6.14, Theorem 7.2 and Section 7.3 in \cite{long2022near}.

\begin{corollary}
There is a deterministic vertex-failure connectivity oracle with $\what{O}(m) + O(d_{\star}m\log^{3} n)$ preprocessing time, $O(m\log^{3} n)$ space, $O(d^{2}(\log^{7} n + \log^{5}n\cdot\log^{4}d))$ update time and $O(d)$ query time.
\label{coro:VFConnOracle}
\end{corollary}

\subsection{A Balanced Sparse Vertex Cut or a Low-Degree Steiner Tree}
\label{sect:BalCutorLDST}

The goal of this subsection is to show \Cref{lemma:CutOrSteinerTree}, a subroutine which given a graph with terminals, outputs either a balanced sparse vertex cut or a low-degree Steiner tree covering a large fraction of terminals. In fact, some expander decomposition algorithms (e.g. \cite{CGLNPS20}) exploit a similar subroutine which either computes a balanced sparse cut or certifies that a large part of the graph is an expander. Our subroutine can be viewed as a weaker and simplified version, because similar to the notion of expanders, a low-degree Steiner tree is also an object that certifies some kind of (weaker) well-linkedness.

At a high level, our algorithm uses a simplified \emph{cut-matching-game} framework. A cut-matching game is an interactive process between a cut player and a matching player with several rounds. Start from a graph with no edge. In each round, the cut player will select a cut and then the matching player is required to add a perfect matching on this cut.  It is known that there exists cut-player strategies against an arbitrary matching player that guarantees the final graph is an expander after $\wtilde{O}(1)$ rounds \cite{KRV09,KKOV07}. In the proof of \Cref{lemma:CutOrSteinerTree}, we show a cut-player strategy that only guarantees the final graph is a \emph{connected graph}. Combining a classic matching player as shown in \Cref{lemma:MatchingPlayer}, we can either find a balanced sparse vertex cut or embed a connected graph covering most of the terminals into the original graph with low vertex congestion. In the latter case, the embedding leads to a low-degree subgraph covering most of the terminals. Finally, picking an arbitrary spanning tree in this subgraph suffices.

Given a cut w.r.t. terminals, \Cref{lemma:MatchingPlayer} will either output a balanced sparse vertex cut or a large matching between terminals that is embeddable into the original graph with low vertex congestion. In fact, it is a simplified version of the matching player in \cite{long2022near}, and we defer its proof to \Cref{sect:MatchingPlayerProof}. 

\begin{lemma}
Let $G$ be an undirected graph with a terminal set $U$. Given a parameter $\phi$ and a partition $(A,B)$ of $U$, there is a deterministic algorithm that computes either
\begin{itemize}
\item a vertex cut $(L,S,R)$ with $|R\cap U|\geq|L\cap U|\geq \min\{|A|,|B|\}/3$ and $|S|\leq \phi\cdot|L\cap U|$, or
\item a matching $M$ between $A$ and $B$ with size $|M|\geq \min\{|A|,|B|\}/3$ s.t. there is an embedding $\Pi_{M\to G}$ of $M$ into $G$ with vertex congestion at most $\lceil 1/\phi\rceil$.
\end{itemize}
The running time is $m^{1+o(1)}$. If the output is a matching $M$, the algorithm can further output the edge set $E(\Pi_{M\to G})$ of the embedding $\Pi_{M\to G}$.
\label{lemma:MatchingPlayer}
\end{lemma}

\begin{lemma}
Let $G$ be an undirected graph with a terminal set $U$. Given parameters $0<\epsilon,\phi\leq 1/4$, there is a deterministic algorithm that computes either
\begin{itemize}
\item a vertex cut $(L,S,R)$ with $|R\cap U|\geq |L\cap U|\geq \epsilon|U|/3$ and $|S|\leq \phi\cdot |L\cap U|$, or
\item a subset $U_{\drop}\subseteq U$ of terminals with $U_{\drop}\leq \epsilon |U|$ and a Steiner tree $\tau$ on $G\setminus U_{\drop}$ of terminal set $U\setminus U_{\drop}$ with maximum degree $O(\log |U|/\phi)$.
\end{itemize}
\label{lemma:CutOrSteinerTree}
\end{lemma}
\begin{proof}
The algorithm is made up of an iteration phase and a postprocessing phase. The iteration phase will maintain an incremental graph $W$ with $V(W) = U$, called the \emph{witness graph}, and its embedding $\Pi_{W\to G}$ into $G$. Precisely, instead of storing the embedding $\Pi_{W\to G}$ explicitly, the algorithm will only store its edge set $E(\Pi_{W\to G})$. Initially, the witness graph $W^{(0)}$ has no edge and $E(\Pi_{W^{(0)}\to G})$ is empty. We use $W^{(i)}$ and $E(\Pi_{W^{(i)}\to G})$ to denote the witness graph and the edge set of the embedding right after the $i$-th round.

In the iteration phase, we do the following steps in the $i$-th round.
\begin{enumerate}
\item We compute all the connected components of $W^{(i-1)}$, which forms a partition ${\cal Q}^{(i-1)}$ of $U$ s.t. each $Q\in {\cal Q}^{(i-1)}$ is a subset of $|U|$, called a \emph{cluster}. If there is a cluster $Q^{\star}\in {\cal Q}^{(i-1)}$ has $|Q^{\star}|\geq (1-\epsilon)|U|$, then we terminate the iteration phase and go to the postprocessing phase, otherwise we proceed to the next step.
\item Because step 1 guarantees that all clusters in ${\cal Q}^{(i-1)}$ have size at most $(1-\epsilon)|U|$, we will partition ${\cal Q}^{(i-1)}$ into two groups ${\cal Q}_{A}$ and ${\cal Q}_{B}$ depending on the following two cases.
\begin{itemize}
\item[(a)] If the all clusters in ${\cal Q}^{(i-1)}$ have size at most $|U|/2$, then we partition ${\cal Q}^{(i-1)}$ into ${\cal Q}_{A}$ and ${\cal Q}_{B}$ s.t. $\sum_{Q\in {\cal Q}_{A}}|Q|\geq |U|/4$ and $\sum_{Q\in {\cal Q}_{B}}|Q|\geq |U|/4$.
\item[(b)] Otherwise, there is a unique cluster $Q^{\star}$ s.t. $|U|/2<|Q^{\star}|\leq (1-\epsilon)|U|$, and we let ${\cal Q}_{A} = {\cal Q}\setminus\{Q^{\star}\}$ and ${\cal Q}_{B} = \{Q^{\star}\}$.
\end{itemize}
Let $A_{i} = \bigcup_{Q\in{\cal Q}_{A}}Q$ and $B_{i} = \bigcup_{Q\in{\cal Q}_{B}}Q$. Note that by definition, $(A_{i},B_{i})$ forms a partition of $U$. We have $|A_{i}|,|B_{i}|\geq |U|/4$  in case (a) and $|A_{i}|,|B_{i}|\geq \epsilon|U|$ in case (b).
\item We apply \Cref{lemma:MatchingPlayer} on graph $G$ and terminal $U$ with parameter $\phi$ and the partition $(A_{i},B_{i})$ of $U$. If we get a vertex cut $(L,S,R)$, it will satisfy $|R\cap U|\geq |L\cap U|\geq \min\{|A_{i}|,|B_{i}|\}/3\geq \epsilon|U|/3$ and $|S|\leq \phi\cdot|L\cap U|$ as desired, so we can terminate the whole algorithm with $(L,S,R)$ as the output. Otherwise, we get a matching $M_{i}$ between $A_{i}$ and $B_{i}$ with size $|M_{i}|\geq |A_{i}|/3$, and the edge set $E(\Pi_{M_{i}\to G})$ of some embedding $\Pi_{M_{i}\to G}$ that has vertex congestion $O(1/\phi)$. Then we let $W^{(i)} = W^{(i-1)}\cup M_{i}$ and $E(\Pi_{W^{(i)}\to G}) = E(\Pi_{W^{(i-1)}\to G})\cup E(\Pi_{M_{i}\to G})$, and proceed to the next round.
\end{enumerate}

If the algorithm does not end at step 3, it exits the iteration phase at step 1, and then we perform the following postprocessing phase. Let $W$ denote the final witness graph with connected components ${\cal Q}$ and a cluster $Q^{\star}\in {\cal Q}$ s.t. $|Q^{\star}|\geq (1-\epsilon)|U|$. Note that $Q^{\star}\subseteq U$. Let $G'$ be the subgraph of $G$ induced by $E(\Pi_{W\to G})$. By the definition of embedding, vertices in $Q^{\star}$ are also connected in $G'$. In other words, $Q^{\star}$ is contained by a connected component of $G'$. We can take an arbitrary spanning tree $\tau$ of this component as a Steiner tree of $V(\tau)\cap U$, and define $U_{\drop} = U\setminus V(\tau)$ be the uncovered terminals. 

We now show that $U_{\drop}$ and $\tau$ have the desire property. The number of uncovered terminals is bounded by $|U_{\drop}|\leq |U|-|V(\tau)\cap U|\leq |U|-|Q^{\star}|\leq \epsilon|U|$, and $\tau$ is a Steiner tree of $U\setminus U_{\drop}$ with maximum degree $O(\log |U|/\phi)$ because $G'$ has maximum degree at most the vertex congestion of $\Pi_{W\to G}$, which is at most $O(\log |U|/\phi)$ by \Cref{claim:NumberOfRounds}.

\begin{claim}
The number of rounds in the iteration phase is at most $O(\log |U|)$, and the vertex congestion of the final embedding $\Pi_{W\to G}$ is at most $O(\log |U|/\phi)$. 
\label{claim:NumberOfRounds}
\end{claim}
\begin{proof}
Note that the early rounds will go into case (a) in step 2, while the late rounds will go into case (b). We bound the number of case-(a) rounds and case-(b) rounds separately.

The number of case-(a) rounds is at most $O(\log |U|)$ by the following reason. We define a potential function $\Phi(W)$ of the witness graph by
\[
\Phi(W) = \sum_{Q\in{\cal Q}}\sum_{v\in Q}\log|Q| = \sum_{Q\in{\cal Q}}|Q|\cdot\log |Q|.
\]
In particular, for each cluster $Q\in{\cal Q}$ and each vertex $v\in Q$, we say the potential at $v$ is $\log|Q|$.

Because initially $\Phi(W^{(0)}) = 0$ and we always have $\Phi(W)\leq |U|\log|U|$, it is sufficient to show that each case-(a) round increases the potential by at least $\Omega(|U|)$. To see this, consider the $i$-th case-(a) round. For each matching edge $\{u,v\}\in M_{i}$, let $Q^{(i-1)}_{v}$ (resp. $Q^{(i-1)}_{u}$) be the connected component of $W^{(i-1)}$ that contains $v$ (resp. $u$), and assume without loss of generality that $|Q^{(i-1)}_{v}|\leq |Q^{(i-1)}_{u}|$. Then the connected component $Q^{(i)}_{v}$ of $W^{(i)}$ that contains $v$ will have $|Q^{(i)}_{v}|\geq 2|Q^{(i-1)}_{v}|$, because $Q^{(i)}_{v}\supseteq Q^{(i-1)}_{u}\cup Q^{(i-1)}_{v}$. In other words, this round will increase the potential at $v$ (from $\log|Q^{(i-1)}_{v}|$ to $\log|Q^{(i)}_{v}|$) by at least $1$. Summing over $|M_{i}|$ matching edges, the total potential $\Phi(W)$ will be increased by at least $|M_{i}|\geq |U|/12$ as desired, because the potential at any $v\in V(W)$ will never drop.

It remains to show that the number of case-(b) rounds is at most $O(\log |U|)$. This is simply because in each round, the matching $M_{i}$ will merge at least $|A_{i}|/3$ terminals in $|A_{i}|$ into the giant cluster $Q^{\star}$, which means $|Q^{\star}|$ will reach the threshold $(1-\epsilon)|U|$ in $O(\log|U|)$ many case-$(b)$ rounds and then the iteration phase ends. 

The final embedding $\Pi_{W\to G}$ has vertex congestion $O(\log|U|/\phi)$ because there are $O(\log|U|)$ rounds and the embedding $\Pi_{M_{i}\to G}$ has vertex congestion $O(1/\phi)$ each round.
\end{proof}

\end{proof}

\subsection{The Low-Degree Steiner Forest Decomposition}
\label{sect:SFDecomp}

\Cref{lemma:SFDecomp} describes the low-degree Steiner forest decomposition algorithm, which invokes \Cref{lemma:CutOrSteinerTree} in a divide-and-conquer fashion. For simplicity, the readers can always assume $\epsilon = 1/2$, which is the value we will choose when constructing the low degree hierarchy. We introduce this tradeoff parameter $\epsilon$ just to show that our algorithm has the same flexibility as those in \cite{duan2020connectivity,long2022near}.

\begin{lemma}
Let $G$ be an undirected graph with a terminal set $U$. Given a parameter $0<\epsilon\leq 1/2$, there is a deterministic algorithm that computes
\begin{itemize}
\item a vertex set $X\subseteq V(G)$, called the separator, s.t. $|X|\leq \epsilon|U|$, and
\item for each connected component $Y$ of $G\setminus X$ s.t. $U$ intersects $V(Y)$, a Steiner tree $\tau_{Y}$ spanning $U \cap V(Y)$ on $Y$ with maximum degree $O(\log^{2}|U|/\epsilon)$.
\end{itemize}
The running time is $m^{1+o(1)}/\epsilon$.
\label{lemma:SFDecomp}
\end{lemma}

\begin{proof}

We will decompose the graph using \Cref{lemma:CutOrSteinerTree} recursively. The algorithm $\SFDecomp(G,U)$ is described in details in \Cref{algo:SFDecomp}.

\begin{algorithm}[H]
\caption{The low-degree Steiner forest decomposition $\SFDecomp(G,U)$\label{algo:SFDecomp}
}
\begin{algorithmic}[1]
\Require An undirected graph $G$ with terminals $U$.
\Ensure A separator $X$ and a collection ${\cal T}$ of Steiner trees $\{\tau_{Y}\}$.

\State Let $\epsilon' = \epsilon/2$ and $\phi = \epsilon'/\log|U|$
\State Apply \Cref{lemma:CutOrSteinerTree} on $G$ and $U$ with parameters $\epsilon'$ and $\phi$.
\If{\Cref{lemma:CutOrSteinerTree} outputs a vertex cut $(L,S,R)$}
\State $(X_{L},{\cal T}_{L})\gets \SFDecomp(G[L],L\cap U)$\label{line:SFDecompCase1}
\State $(X_{R},{\cal T}_{R})\gets \SFDecomp(G[R],R\cap U)$
\State Return $X = X_{L}\cup X_{R}\cup S$ and ${\cal T} = {\cal T}_{L}\cup {\cal T}_{R}$.
\Else
\State Otherwise \Cref{lemma:CutOrSteinerTree} outputs $U_{\drop}\subseteq U$ and a Steiner tree $\tau$ of $U\setminus U_{\drop}$ on $G\setminus U_{\drop}$.\label{line:SFDecompCase2}
\State Return $X = U_{\drop}$ and ${\cal T} = \{\tau\}$.
\EndIf
\end{algorithmic}
\end{algorithm}

We first show the correctness of the collection ${\cal T}$ of Steiner trees outputted by $\SFDecomp(G,U)$. 

\begin{itemize}
\item If \Cref{lemma:CutOrSteinerTree} outputs $U_{\drop}$ and $\tau$ (i.e. the algorithm reaches \Cref{line:SFDecompCase2}), then \Cref{lemma:CutOrSteinerTree} guarantees that $\tau$ is a Steiner tree of $U\setminus U_{\drop}$ on $G\setminus U_{\drop}$. We point out that $G\setminus U_{\drop}$ can possibly be a disconnected graph, but terminals $U\setminus U_{\drop}$ must belong to a same connected component $Y$ of $G\setminus U_{\drop}$. Hence, $Y$ is the only connected component of $G\setminus X$ (because we set $X = U_{\drop}$) s.t. $V(Y)$ intersects $U\setminus X$, and $\tau$ is a valid Steiner tree spanning $(U\setminus X)\cap V(Y)$ on $Y$ with maximum degree $O(\log|U|/\phi) = O(\log^{2}|U|/\epsilon)$. 

\item Now consider the other case that \Cref{lemma:CutOrSteinerTree} outputs a vertex cut $(L,S,R)$. We assume inductively that the two recursions of $G_{L}$ and $G_{R}$ output the correct ${\cal T}_{L}$ and ${\cal T_{R}}$. Furthermore, we have that the collection of connected components of $G\setminus X$ are the union of those of $G_{L}\setminus X_{L}$ and $G_{R}\setminus X_{R}$ by the property of vertex cut and our definition of $X$, which directly leads to the correctness of ${\cal T}$.

\end{itemize}

We next bound the size of the separator $X$. By the algorithm, the separator is the union of the sets $U_{\drop}$ of all \emph{leaf steps} (i.e. the steps that reach \Cref{line:SFDecompCase2}) and the cuts $S$ of all \emph{intermediate steps} (i.e. the steps that reach \Cref{line:SFDecompCase1}). Because the terminal sets of the leaf steps are mutually disjoint, the total size of $U_{\drop}$ are at most $\epsilon'|U|$ from the guarantee of \Cref{lemma:CutOrSteinerTree}. To bound the total size of $S$ of the intermediate steps, we will use a charging argument. Each intermediate step $\SFDecomp(G',U')$ will have $|S|\leq \phi|L\cap U'|$ and $|L\cap U'|\leq |U'|/2$ (from $|L\cap U'|\leq |R\cap U'|$) guaranteed by \Cref{lemma:CutOrSteinerTree}, and then we charge $\phi$ units to each terminals of $L\cup U'$. In the whole charging process, observe that each terminal $u\in U$ will be charged at most $\log|U|$ times ($\phi$ units each time), because from top to bottom, each time $u$ is charged, it will go to the left branch with number of terminals halved. Therefore, the total size of $S$ over all intermediate steps is at most $\log|U|\cdot\phi\leq \epsilon'|U|$. We can conclude that $|X|\leq \sum|U_{\drop}| + \sum|S|\leq 2\epsilon'|U|=\epsilon|U|$ as desired.

Lastly, we analyse the running time. Observe that at each intermediate step $\SFDecomp(G',U')$, no matter we go into the left or right branch, the number of terminals is decreased by a factor of $1-\epsilon'/3$, because \Cref{lemma:CutOrSteinerTree} guarantees that $|R\cap U'|\geq |L\cap U'|\geq \epsilon'|U'|/3$. Hence the maximum depth of the recursion is at most $O(\log |U|/\epsilon')$. The total running time is then $m^{1+o(1)}/\epsilon$, because the total graph size of steps with the same depth is at most $m$ and the running time of each step is almost-linear by \Cref{lemma:CutOrSteinerTree}.
\end{proof}

As shown in \cite{duan2020connectivity}, to construct a low-degree hierarchy $({\cal C},{\cal T})$, it suffices to invoke the low-degree Steiner forest decomposition (with $\epsilon = 1/2$) $O(\log n)$ times. The algorithm is shown in \Cref{algo:LowDegreeHierarchy}, and for completeness, the proof of correctness is included in \Cref{sect:CorrectnessLowDegreeHierarchy}.

\begin{algorithm}[h]
\caption{The construction of the low-degree hierarchy\label{algo:LowDegreeHierarchy}
}
\begin{algorithmic}[1]
\Require An undirected graph $G$.
\Ensure A low-degree hierarchy $({\cal C}, {\cal T})$.

\State Initialize $i = 1$, $X_{1} = V(G)$.
\While{$X_{i}$ is not empty}
\State $(X_{i+1},{\cal T}_{i})\gets \SFDecomp(G,X_{i})$ with $\epsilon = 1/2$.
\State $i\gets i+1$.
\EndWhile
\State $p\gets i-1$, which denotes the number of levels.
\For{each level $i$}
\State $U'_{i} \gets X_{i}\cup ...\cup X_{p}$.
\State ${\cal C}_{i}\gets$ the connected component of $G\setminus U'_{i+1}$ (particularly, $U'_{p+1} = \emptyset$).
\State $U_{i}\gets U'_{i}\setminus U'_{i+1}$, which denotes the terminals of level $i$.
\EndFor
\end{algorithmic}
\end{algorithm}

\section{The Preprocessing Algorithm}
\label{sect:Preprocessing}

In this section, we will describe the preprocessing algorithm, which basically first computes the low degree hierarchy on $G_{\on}:=G[V_{\on}]$, and then constructs some affiliated data structures on top of the hierarchy.

The low degree hierarchy $({\cal C}, {\cal T})$ is computed by applying \Cref{thm:LowDegreeHierarchy} on $G_{\on}$, if $G_{\on}$ is a connected graph. In the case that $G_{\on}$ is not connected, we simply apply \Cref{thm:LowDegreeHierarchy} on each of the connected component of $G_{\on}$. To simplify the notations, we use $({\cal C}, {\cal T})$ to denote the union of hierarchies of connected components of $G_{\on}$, and still say $({\cal C}, {\cal T})$ is the low degree hierarchy of $G_{\on}$. Note that $({\cal C},{\cal T})$ has all properties in \Cref{def:LowDegreeHierarchy}, except that the top level ${\cal C}_{1}$ are now made up of connected components of $G_{\on}$.

In \Cref{sect:ArtificialGraph}, we introduce the notions of artificial edges and the artificial graph $\hat{G}$. In \Cref{sect:GlobalOrder}, we define a global order $\pi$ based on Euler tour orders of Steiner trees in ${\cal T}$, and then construct a 2D range counting structure which can answers the number of edges in $E(\hat{G})$ between two intervals on $\pi$. Finally, in \Cref{sect:PreprocessingAnalyse}, we summarize what we will store, and analyse the preprocessing time and the space complexity.

\subsection{Artificial Edges and the Artificial Graph $\hat{G}$}
\label{sect:ArtificialGraph}

The artificial graph $\hat{G}$ is a multi-graph constructed by adding some \emph{artificial edges} into the original graph $G$ in the following way. For each component $\gamma\in{\cal C}$, let $A_{\gamma}$ collect the neighbors of $V(\gamma)$ in $G$, formally defined by $A_{\gamma} = \{v\mid v\in V(G)\setminus V(\gamma)\text{ s.t. }\exists \{u,v\}\in E(G)\text{ 
with }u\in V(\gamma)\}$.
We call $A_{\gamma}$ the \emph{adjacency list} of $\gamma$. Let $A_{\gamma,\on} = A_{\gamma}\cap V_{\on}$ and $A_{\gamma,\off} = A_{\gamma}\cap V_{\off}$. 
Next, we let $B_{\gamma,\off} = A_{\gamma,\off}$ and let $B_{\gamma,\on}$ be an arbitrary subset of $A_{\gamma,\on}$ with size $\min\{d_{\star}+1, |A_{\gamma,\on}|\}$. Then define $B_{\gamma} = B_{\gamma,\on}\cup B_{\gamma,\off}$. 

The artificial edges added by the component $\gamma$ is then $\hat{E}_{\gamma} = \{\{u,v\}\mid u\in A_{\gamma}, v\in B_{\gamma}, u\neq v\}$.
Namely, $\hat{E}_{\gamma}$ consists of a clique on $B_{\gamma}$ and a biclique between $B_{\gamma}$ and $A_{\gamma}\setminus B_{\gamma}$. Finally, the artificial graph $\hat{G}$ is defined by $\hat{G} = G +\sum_{\gamma\in{\cal C}}\hat{E}_{\gamma}$.
We emphasize that $\hat{G}$ is a multi-graph, and those edges connecting the same endpoints will have different identifiers.

We show some useful properties in \Cref{prop:5.1}. \Cref{prop:RobustB} of \Cref{prop:5.1} basically says that, if $A_{\gamma}$ has an on-vertex after update, then $B_{\gamma}$ also has one.

\begin{proposition}
We have the following.
\begin{enumerate}
\item\label{prop:SizeOfA}
 $\sum_{\gamma\in{\cal C}} |A_{\gamma}|\leq O(pm)$.
\item\label{prop:NumberOfAllEdges} $|E(\hat{G})|\leq O(pm(n_{\off} + d_{\star}))$.
\item\label{prop:RobustB} Given any update $D\subseteq V$ with $|D|\leq d_{\star}$, if $(A_{\gamma,\on}\setminus D)\cup(A_{\gamma,\off}\cap D)\neq \emptyset$, then we have $(B_{\gamma,\on}\setminus D)\cup (B_{\gamma,\off}\cap D)\neq\emptyset$.
\end{enumerate}
\label{prop:5.1}
\end{proposition}
\begin{proof}
\underline{Part 1.} For each $\gamma\in {\cal C}$, observe that $|A_{\gamma}|\leq \sum_{v\in V(\gamma)}\deg_{G}(v)$. Hence $\sum_{\gamma\in{\cal C}}|A_{\gamma}|\leq O(pm)$ because each vertex can appear in at most $p$ components (at most one at each level).

\underline{Part 2.} By definition, $|E(\hat{G})|\leq m + \sum_{\gamma\in C} |\hat{E}_{\gamma}|\leq m + \sum_{\gamma\in{\cal C}}|A_{\gamma}|\cdot|B_{\gamma}|$. Note that $|B_{\gamma}|\leq n_{\off} + d_{\star} + 1$ for all $\gamma\in{\cal C}$ by construction. Combining part 1, we have $|E(\hat{G})|\leq O(pm(n_{\off} + d_{\star}))$.

\underline{Part 3.} If $|A_{\gamma,\on}|\leq d_{\star}+1$, we have $A_{\gamma} = B_{\gamma}$ by construction and the proposition trivially holds. Otherwise, $B_{\gamma}$ will include $d_{\star}+1$ vertices in $A_{\gamma,\on}$. Because $|D|\leq d_{\star}$, at least one of them will survive in $B_{\gamma,\on}\setminus D$, which implies $(B_{\gamma,\on}\setminus D)\cup (B_{\gamma,\off}\cap D)\neq\emptyset$. 
\end{proof}

\subsection{The Global Order and Range Counting Structures}
\label{sect:GlobalOrder}

Next, we define an order, called the \emph{global order} and denoted by $\pi$, over the whole vertex set $V(G)$, based on the Euler Tour orders of Steiner trees in ${\cal T}$.

For each $\tau\in{\cal T}$, we define its Euler tour order $\ET(\tau)$ as an ordered list of vertices in $V(\tau)$ ordered by the time stamps of their first appearances in an Euler tour of $\tau$ (starting from an arbitrary root). Intuitively, the Euler tour order $\ET(\tau)$ can be interpreted as a linearization of $\tau$, i.e. after the removal of failed vertices in $\tau$, the remaining subtrees will corresponding to intervals on $\ET(\tau)$, as shown in \Cref{lemma: Interval}.

\begin{lemma}[Lemma 6.3 in \cite{long2022near}, Rephrased]

    Let $\tau$ be an undirected tree with maximum vertex degree $\Delta$. A removal of d failed vertices from $\tau$ will split $\tau$ into at most $O(\Delta d)$ subtrees $\hat{\tau}_1, \hat{\tau}_2,..., \hat{\tau}_{O(\Delta d)}$, and there exists a set ${\cal I}_{\tau}$ of at most $O(\Delta d)$ disjoint intervals on $\text{ET}(\tau)$, such that each interval is owned by a unique subtree and for each subtree $\tau_i$, $V (\tau_i)$ is equal to the union of intervals it owns.

    Furthermore, by preprocessing $\tau$ in $O(|V (\tau)|)$ time, we can store $\ET(\tau)$ and some additional information in $O(|V(\tau)|)$ space, which supports the following operations.
    \begin{itemize}
    \item Given a set $D_{\tau}$ of $d$ failed vertices, the intervals ${\cal I}_{t}$ can be computed in $O(\Delta d \log (\Delta d))$ update time.
    \item Given a vertex $v\in V(\tau)\setminus D_{\tau}$, it takes $O(\log d)$ query time to find an interval $I\in{\cal I}_{t}$ s.t. vertices in $I$ are connected to $v$ in $\tau\setminus D_{\tau}$.
    \end{itemize}

    \label{lemma: Interval}
\end{lemma}

Given the Euler tour orders of all $\tau\in{\cal T}$, we define the global order $\pi$ as follows. We first  concatenate $\ET(\tau)\cap U(\tau)$ (i.e. the restriction of $\ET(\tau)$ on the terminals of $\tau$) of all $\tau\in{\cal T}$ in an arbitrary order, and then append all vertices in $V_{\off}$ to the end. Recall that $\{U(\tau)\mid \tau\in{\cal T}\}$ partitions $V_{\on}$, so $\pi$ is well-defined.

With the global order $\pi$, we will construct a 2D-range counting structure $\Table$, which can answer the number of edges in $E(\hat{G})$ that connecting two disjoint intervals on $\pi$. We first initialize $\Table_{\init}$ to be an ordinary 2D array on range $\pi\times \pi$. For each $u,v\in \pi$, we store a non-negative integer in the entry $\Table_{\init}(u,v)$ representing the number of edges in $E(\hat{G})$ connecting vertices $u$ and $v$.

\begin{lemma}
Suppose that we can access the lists 
$A_{\gamma}$ and $B_{\gamma}$ for all $\gamma\in{\cal C}$. There is an combinatorial algorithm that computes $\Table_{\init}$ in $O(|E(\hat{G})|)$ time, or $\Table_{\init}$ can be computed in $O(p\cdot n^{\omega})$ time using fast matrix multiplication.
\label{lemma:InitializeTable}
\end{lemma}

\begin{proof}
A trivial construction of $\Table_{\init}$ is to construct the edge sets $E(\hat{G})$ explicitly, and then scan the edges one by one. Obviously, this takes $O(|E(\hat{G})|)$ time.

When $|E(\hat{G})|$ is large, we can use fast matrix multiplication (FMM) to speed up the construction of $\Table_{\init}$. Recall that $E(\hat{G}) = E(G) + \sum_{\gamma\in{\cal C}}\hat{E}_{\gamma}$. We first add the contribution of $E(G)$ into $\Table_{\init}$ using the trivial algorithm, which takes $O(m)$ time. Next, we compute the contribution of artificial edges, i.e. $\sum_{\gamma\in{\cal C}}\hat{E}_{\gamma}$, using FMM. We construct a matrix $X$ with $n$ rows and $|{\cal C}|$ columns, where rows are indexed by the global order $\pi$ and columns are indexed by components (in an arbitrary order). For each vertex $u\in \pi$ and component $\gamma\in{\cal C}$, the entry $X(u,\gamma) = 1$ if and only if $u\in A_{\gamma}$. Similarly, we define an $n$-row $|{\cal C}|$-column matrix $Y$, in which each entry $Y(u,\gamma)=1$ if and only if $u\in A_{\gamma}\setminus B_{\gamma}$. Let $Z = X\cdot X^\intercal - Y\cdot Y^{\intercal}$. Observe that, for each pair of distinct vertices $u,v\in \pi$,
\begin{align*}
Z(u,v) &= \sum_{\gamma\in{\cal C}}(X(u,\gamma)\cdot X(v,\gamma) - Y(u,\gamma)\cdot Y(v,\gamma))\\
&= \sum_{\gamma\in{\cal C}}\mathds{1}[u,v\in A_{\gamma}] - \mathds{1}[u,v\in A_{\gamma}\setminus B_{\gamma}]\\
&= \sum_{\gamma\in{\cal C}}\mathds{1}[\{u,v\}\in \hat{E}_{\gamma}].
\end{align*}
Therefore, the matrix $Z$ count the contribution of $\sum_{\gamma}\hat{E}_{\gamma}$ correctly and the last step is to add $Z$ to $\Table_{\init}$. The construction time is dominated by the computation of $Z$, which takes $O(p\cdot n^{\omega})$ time because it involves multiplying an $n\times |{\cal C}|$ matrix and a $|{\cal C}|\times n$ matrix, and $|{\cal C}| = O(pn)$.

\end{proof}

\begin{lemma}
With access to the positive entries of $\Table_{\init}$, we can construct a data structure $\Table$ that given any disjoint intervals $I_{1}$ and $I_{2}$ on $\pi$, answers in $O(\log n)$ time the number of edges in $E(\hat{G})$ with one endpoint in $I_{1}$ and the other one in $I_{2}$. The structure $\Table$ can be constructed in $O(N\log n)$ time and takes space $O(N\log n)$, where $N$ denotes the number of positive entries in $\Table_{\init}$.
\label{lemma:2DRangeCounting}
\end{lemma}
\begin{proof}
We simply construct $\Table$ as a standard weighted 2D range counting structure of $\Table_{\init}$, By using textbook algorithms such as range trees and persistent segment trees, we can construct $\Table$ in $O(N\log n)$ time and it takes space $O(N\log n)$. The correctness of $\Table$ follows the definition of $\Table_{\init}$. 
\end{proof}

\subsection{Preprocessing Time and Space Analysis}
\label{sect:PreprocessingAnalyse}

In conclusion, we will compute and store the following in the preprocessing phase.

\begin{itemize}
\item 
First, we store the low degree hierarchy $({\cal C}, {\cal T})$. Constructing the low degree hierarchy takes $\hat{O}(m)$ time by \Cref{thm:LowDegreeHierarchy}. Storing the low degree hierarchy explicitly takes $O(pn)$ space, because for each level $i$, the components in ${\cal C}_{i}$ are vertex disjoint, also Steiner trees in ${\cal T}_{i}$.

\item
Next, for each $\gamma\in{\cal C}$, we store the lists $A_{\gamma}$ and $B_{\gamma}$ after ordering them by $\pi$. Computing the lists $A_{\gamma}$ and $B_{\gamma}$ takes $O(pm)$ time by checking the incident edges of each vertex in each component. Storing the lists $A_{\gamma}$ and $B_{\gamma}$ takes $O(pm)$ space by \Cref{prop:SizeOfA} in \Cref{prop:5.1}. Additionally, for each $\gamma\in{\cal C}$, store the list $A_{\gamma,\on}$.

For each $v\in V_{\off}$ and $\gamma\in{\cal C}$, store a binary indicator to indicate whether $v\in A_{\gamma,\off}$ or not. Computing the indicators takes $O(pm)$ time by scanning all the lists $A_{\gamma}$. Storing the indicators explicitly takes $O(|V_{\off}|\cdot|{\cal C}|) = O(pn\cdot |V_{\off}|)$ space.

\item We also store the global order $\pi$, which takes $O(n)$ space. For each $\tau\in {\cal T}$, we store $\ET(\tau)$ and the additional information stated in \Cref{lemma: Interval} in $O(|V(\tau)|)$ space. Computing the things above takes totally $\sum_{\tau\in {\cal T}}|V(\tau)| = O(pn)$ time by \Cref{lemma: Interval}.

\item Finally, we store the data structure $\Table$. Combining \Cref{prop:NumberOfAllEdges} in \Cref{prop:5.1} and \Cref{lemma:InitializeTable,lemma:2DRangeCounting}, we can compute $\Table$ in $O(pm(n_{\off}+d)\log n)$ time using an combinatorial algorithm, or in $O(p\cdot n^{\omega}\log n)$ time using fast matrix multiplication. The space to store $\Table$ is $\min\{pm(n_{\off} + d)\log n, n^{2}\}$.
\end{itemize}

In conclusion, the total preprocessing time can be upper bounded by $\what{O}(m) + O(pm(n_{\off}+d)\log n)$ using an combinatorial algorithm, then $t_{p}=\hat{\Omega}(md)$, or $\hat{O}(m) + O(p\cdot n^{\omega}\log n)$ using fast matrix multiplication. The space complexity is $O(\min\{pm(n_{\off} + d)\log n,n^{2}\})$. Because the low degree hierarchy has $p=O(\log n)$ levels, the preprocessing time is $\what{O}(m) + O(\min\{m(n_{\off} + d)\log^{2} n,n^{\omega}\log^{2} n\})$, and the space is $O(\min\{m(n_{\off} + d)\log^{2} n,n^{2}\})$.

\section{The Update Algorithm}
\label{sect:UpdateQuery}

Let $D\subseteq V(G)$ be a given update. We use $D_{\on} = D\cap V_{\on}$ to denote the the vertices that will be turned off in this update and $D_{\off} = D\cap V_{\off}$ to denote the vertices that will be turned on. Let $V_{\new} = (V_{\on}\setminus D_{\on})\cup D_{\off}$ be the on-vertices after updates. 

Our update strategy is to recompute the connectivity of a subset of \emph{affected vertices} $Q^{\star}\subseteq V_{\new}$ on some affected graph $G^{\star}$. In \Cref{sect:AffectedGraph}, we will define $Q^{\star}$ and $G^{\star}$, and prove that $Q^{\star}$ has the same connectivity on the affected graph $Q^{\star}$ and the updated original graph $G[V_{\new}]$. In \Cref{sect:SolvingIntervals}, we will partition $G^{\star}$ into a small number of sets s.t. each set forms an \emph{interval} on the global order $\pi$ and it is certified to be connected by some Steiner tree in ${\cal T}$. Thus, it suffices to solve the connectivity of intervals on $Q^{\star}$, which is formalized in \Cref{lemma:Boruvka}. In \Cref{sect:QueryAlgorithm}, we will discuss how to answer the connectivity of $u$ and $v$ on $V_{\new}$ by reducing it to a query on $Q^{\star}$.

\begin{theorem}
There exists a deterministic fully dynamic sensitivity oracle for subgraph connectivity with  $O(\min\{m(n_{\off} + d_{\star})\log^{2}n,n^{2}\})$ space, $O(d^{2}\log^{7} n)$ update time and $O(d)$ query time. The preprocessing time is $\what{O}(m) + O(m(n_{\off} + d_{\star})\log^{2} n)$ by a combinatorial algorithm, and $\what{O}(m) + O(n^{\omega}\log^{2} n)$ using fast matrix multiplication.
\label{thm:FullyDynamicOracle}
\end{theorem}

We first conclude our fully dynamic sensitivity oracle for subgraph connectivity in \Cref{thm:FullyDynamicOracle}. The bounds on preprocessing time and space are shown in \Cref{sect:PreprocessingAnalyse}. The update time is given by \Cref{lemma:Boruvka}. The query time is discussed in \Cref{sect:QueryAlgorithm}.

\subsection{Affected Vertices $Q^{\star}$ and the Affected Graph $G^{\star}$}
\label{sect:AffectedGraph}

For each component $\gamma\in{\cal C}$, we call $\gamma$ an \emph{affected component} if $V(\gamma)$ intersects $D_{\on}$, otherwise it is \emph{unaffected}. Let ${\cal C}_{\aff}$ denote the set of affected components. Let ${\cal T}_{\aff} = \{\tau(\gamma)\mid \gamma\in{\cal C}_{\gamma}\}$ denote the Steiner trees corresponding to affected components.

We then define the \emph{affected vertices} to be $Q^{\star} = D_{\off}\cup\bigcup_{\tau\in{\cal T}_{\aff}} U(\tau) \setminus D_{\on}$. 
Namely, $Q^{\star}$ collect the newly opened vertices and the open terminals of affected components. Note that $Q^{\star}\subseteq V_{\new}$. The \emph{affected graph} $G^{\star}$ is $G^{\star} = \hat{G}[Q^{\star}] - \sum_{\gamma\in{\cal C}_{\aff}}\hat{E}_{\gamma}$.
In other words, $G^{\star}$ is the subgraph of the artificial graph $\hat{G}$ induced by the affected vertices $Q^{\star}$, with the artificial edges from affected components removed.

\begin{lemma}
For any two vertices $u,v\in Q^{\star}$, $u$ and $v$ are connected in $G[V_{\new}]$ if and only if $u$ and $v$ are connected in $G^{\star}$.
\label{lemma:ConnEq}
\end{lemma}

\begin{proof}
First, we show that, if $u$ and $v$ are connected in $G^{\star}$, then they are connected in $G[V_{\new}]$. It is sufficient to show that, for each edge $\{u,v\}\in E(G^{\star})$, there exists a path $P$ in $G[V_{\new}]$ connecting $u$ and $v$. If the edge $\{u,v\}$ is not an artificial edge, the path $P$ trivially exists because $\{u,v\}$ is an original edge in $G$ and $u,v\in Q^{\star}\subseteq V_{\new}$. Otherwise, by the definition of $G^{\star}$, there exists an unaffected component $\gamma$ s.t. $\{u,v\} \in \hat{E}_{\gamma}$, which means $u,v\in A_{\gamma}$ by the construction of $\hat{E}_{\gamma}$. Then there exists a path $P$ connecting $u$ and $v$ with all internal vertices falling in $\gamma$. Moreover, since $\gamma$ is unaffected, all internal vertices of $P$ are inside $V_{\new}$, so $P$ is path in $G[V_{\new}]$ connecting $u$ and $v$ as desired.

Next, We show that $u$ and $v$ are connected in $G^{\star}$, assuming that they are connected in $G[V_{\new}]$. Let $P$ be a simple path in $G[V_{\new}]$ connecting $u$ and $v$. We can write $P_{uv}$ in the form $P_{1}\circ P_{2}\circ ... \circ P_{\ell}$ s.t. each subpath $P_{j}$ has endpoint $x_{j},y_{j}\in Q^{\star}$ and the internal vertices of $P_{j}$ are disjoint from $Q^{\star}$. Note that $x_{1} = u, y_{\ell} = v$ and $y_{j} = x_{j-1}$ for each $1\leq j\leq \ell-1$. It suffices to show that $x_{j}$ and $y_{j}$ are connected in $G^{\star}$ for all $1\leq j\leq \ell$.

We consider some $j$ and write $x=x_{j}$, $y=y_{j}$ and $P_{xy} = P_{j}$ to simplify the notations. If $P_{xy}$ has only one edge, then this is an original edge and it appears in $G^{\star}$, which means $x$ and $y$ are connected in $G^{\star}$. From now suppose $P_{xy}$ has internal vertices. Let $\hat{\gamma}\in{\cal C}$ be the maximal unaffected component including all internal vertices of $P_{xy}$ ($\hat{\gamma}$ must exist by \Cref{claim:MinimalComponent}).

\begin{claim}
Let $\gamma\in{\cal C}$ be the minimal component including all internal vertices of $P_{xy}$. Then $\gamma$ is an unaffected component.
\label{claim:MinimalComponent}
\end{claim}
\begin{proof}
We will prove the following statement instead: at least one internal vertex $w$ of $P_{xy}$ falls in $U(\gamma)$. This is sufficient because $w\notin Q^{\star}$ by the construction of $P_{xy}$, and then the component $\gamma$ with $w\in U(\gamma)$ cannot be affected (if $\gamma$ is affected, $w\in U(\gamma)\subseteq U(\tau(\gamma))$ will be include in $Q^{\star}$, a contradiction).

Let $\gamma_{1},\gamma_{2},...,\gamma_{k}$ denote the child-components of $\gamma$. Note that $V(\gamma)$ can be partitioned into $U(\gamma), V(\gamma_{1}), V(\gamma_{2}),...,V(\gamma_{k})$. Because $\gamma$ is the minimal component including all internal vertices of $P_{xy}$. Then either (1) all internal vertices of $P_{xy}$ fall in $U(\gamma)$, or (2) there exists two adjacent internal vertices $a,b\in P_{xy}$ belonging to two different parts in the above partition. If case (1) holds, the statement is trivially correct. Suppose case (2) holds. Then either $a\in U(\gamma)$ or $b\in U(\gamma)$, because $a$ and $b$ are connected by an original edge in $G_{\on}$ but there is no edge in $G_{\on}$ connecting two child-components of $\gamma$ by property (1) in \Cref{def:LowDegreeHierarchy}.
\end{proof}

We have $B_{\hat{\gamma}}\cap V_{\new}\subseteq A_{\hat{\gamma}}\cap V_{\new}\subseteq Q^{\star}$ by the definition of $Q^{\star}$ and the fact that $\hat{\gamma}$ has an affected parent-component. Besides, we have $x,y\in A_{\hat{\gamma}}\cap V_{\new}$. By \Cref{prop:RobustB} in \Cref{prop:5.1}, there is a vertex $w\in B_{\hat{\gamma}}\cap V_{\new}$. Therefore, in the artificial graph $\hat{G}$, there exists two artificial edges $\{x,w\}$ and $\{w,y\}$ added by $\gamma$. These two edges remain in $G^{\star}$ because $x,y,w\in Q^{\star}$ and $\gamma$ is unaffected, which means $x$ and $y$ are connected in $Q^{\star}$.

\end{proof}

\subsection{Solving Connectivity of Intervals}
\label{sect:SolvingIntervals}

Although the primary goal of our update algorithm is to compute the connectivity of $Q^{\star}$ on $G[V_{\new}]$, \Cref{lemma:ConnEq} tells that it is equivalent to compute the connectivity of $Q^{\star}$ on $G^{\star}$.

\begin{lemma}
There is a deterministic algorithm that computes a partition ${\cal I}$ of $Q^{\star}$ s.t. each set $I\in{\cal I}$ forms an interval on $\pi$ and all vertices in ${\cal I}$ are connected in $G^{\star}$, and then computes a partition ${\cal R}$ of ${\cal I}$ s.t. for each group $R\in{\cal R}$, the union of intervals in $R$ forms a (maximal) connected component of $G^{\star}$. The running time is $O(p^{2}d^{2}\Delta^{2}\log n)$.
\label{lemma:Boruvka}
\end{lemma}

\paragraph{Intervals.} We first describe how to compute the partition ${\cal I}$ of $Q^{\star}$. Because we require each set $I\in{\cal I}$ forms an \emph{interval} on the global order $\pi$, we can represent $I$ by the positions of its endpoints on $\pi$. Recall that $Q^{\star} = (\bigcup_{\gamma\in{\cal C}_{\aff}}U(\gamma)\setminus D_{\on})\cup D_{\off}$. 
\begin{itemize}
\item We first construct the intervals of $\bigcup_{\tau\in{\cal T}_{\aff}}U(\tau)\setminus D_{\on}$ by exploiting the Steiner trees. For each $\tau\in {\cal T}_{\aff}$, by invoking \Cref{lemma: Interval} on $\tau$ with failed vertices $D_{\on}$, we will obtain a partition ${\cal I}'_{\tau}$ of $V(\tau)\setminus D_{\on}$ s.t. each $I'\in{\cal I}'_{\tau}$ is an interval on $\ET(\tau)$ and it is contained by a subtree of $\tau\setminus D_{\on}$. We construct a set ${\cal I}_{t}$ of intervals on $\ET(\tau)\cap U(\tau)$ by taking the restriction of intervals ${\cal I}'_{t}$ on $U(\tau)$. Therefore, intervals in ${\cal I}_{t}$ are indeed intervals on $\pi$ because $\ET(\tau)\cap U(\tau)$ is a consecutive sublist of $\pi$. Also, for each interval $I\in{\cal I}_{\tau}$, vertices in $I$ are connected in $G[V_{\new}]$ (because $\tau\setminus D_{\on}$ is a subgraph of $G[V_{\new}]$), which implies vertices in $I$ are connected in $G^{\star}$ by \Cref{lemma:ConnEq}.

\item For each vertex $v\in D_{\off}\subseteq Q^{\star}$, we construct a singleton interval $I_{v} = \{v\}$.

\end{itemize}
Finally, the whole set of intervals is ${\cal I} = \bigcup_{\tau\in {\cal T}_{\aff}} {\cal I}_{\tau}\cup \{I_{v}\mid v\in D_{\off}\}$. 

\begin{proposition}
The total number of intervals is $|{\cal I}| = O(pd\Delta)$, and computing all intervals takes $O(pd\Delta\log(d\Delta))$ time.
\label{prop:ConstructI}
\end{proposition}

\begin{proof}
By \Cref{lemma: Interval}, the number of intervals generated by a tree ${\tau}\in {\cal T}_{\aff}$ is at most $O(|V(\tau)\cap D_{\on}|\cdot\Delta)$, and it takes $O(|V(\tau)\cap D_{\on}|\cdot\Delta\cdot \log(|V(\tau)\cap D_{\on}|\cdot\Delta))$ time to generate them. Observe that $\sum_{\tau\in{\cal T}_{\aff}} = O(p\cdot|D_{\on}|)$ because each vertex in $D_{\on}$ can appear in at most $p$ trees in ${\cal T}$ (at most one at each level). Furthermore, the trivial intervals generated by vertices in $D_{\off}$ is obviously $|D_{\off}|$. Therefore, the total number of intervals in $O(p|D_{\on}|\Delta) + |D_{\off}| = O(pd\Delta)$, and computing all intervals takes $O(pd\Delta(d\Delta))$ time.
\end{proof}

\paragraph{\Boruvka's Algorithm.} We now discuss how to compute the partition ${\cal R}$ of ${\cal I}$. We will merge the intervals by a \Boruvka's styled algorithm. The algorithm has several \emph{phases}, and each phase $j$ receive a partition ${\cal R}^{(j)}$ of ${\cal I}$ as input. each group $R\in {\cal R}^{(j)}$ is either \emph{active} or \emph{inactive}. Initially, ${\cal R}^{(1)} = \{\{I\}\mid I\in {\cal I}\}$ is the trivial partition of ${\cal I}$ and all groups in ${\cal R}^{(1)}$ are active. For each phase $j$, we do the following to update ${\cal R}^{(j)}$ to ${\cal R}^{(j+1)}$.
\begin{enumerate}
\item For each active group $R$ in ${\cal R}^{(j)}$, we will ask the following \emph{adjacency query}.
\begin{itemize}
\item[(Q1)] Given an active group $R\in{\cal R}^{(j)}$, find another active group $R'\in{\cal R}^{(j)}$ s.t. there exists an edge $e=\{u,v\}\in E(G^{\star})$ with $u\in I_{u}\in R$ and $v\in I_{v}\in R'$, or claim that there is no such $R'$.
\end{itemize}
After asking (Q1) for all active groups, for each active group $R$, if (Q1) tells that no such $R'$ exists, we mark $R$ as an inactive group, otherwise we find an \emph{adjacent group-pair} $\{R,R'\}$.

\item Given the adjacent group-pairs in step 1, we construct a graph $K$ with vertices corresponding to active groups and edges corresponding to adjacent group-pairs. Note that for each adjacent group-pair $\{R,R'\}$, $R$ and $R'$ must still be active. Then, for each connected component of $K$, we merge the groups inside it into a new active group. 
\end{enumerate}
The algorithm terminates once it reach a phase $\bar{j}$ s.t. all groups in ${\cal R}^{(\bar{j})}$ are inactive, and we let ${\cal R} = {\cal R}^{(\bar{j})}$ be the final output. Obviously, ${\cal R}$ satisfies the output requirement of \Cref{lemma:Boruvka}. Furthermore, the number of phases is bounded in \Cref{prop:BoruvkaPhasesNumber}. Let ${\cal R}^{(j)}_{\act}\subseteq {\cal R}^{(j)}$ denote the active groups in ${\cal R}^{(j)}$ at the moment when phase $j$ starts and let $\bar{k}^{(j)} = {\cal R}^{(j)}_{\act}$.

\begin{proposition}
For each $j\geq 2$, $\bar{k}^{(j)}\leq \bar{k}^{(j-1)}/2$. The number of phases is $O(\log |{\cal I}|)$.
\label{prop:BoruvkaPhasesNumber}
\end{proposition}
\begin{proof}
At each phase, the number of active groups is halved because we mark all old active group without adjacent group inactive in step 1, and each connected component of the graoh $K$ in step 2 contains at least two old active groups. Because initially $\bar{k}^{(0)} = |{\cal R}^{(0)}| = |{\cal I}|$, the number of phases is $O(\log |{\cal I}|)$.
\end{proof}

Next, we will discuss the implementation of step 1. Basically, for each phase $j$, we need an algorithm that answers the \emph{adjacency query} (Q1) efficiently. Instead of answering (Q1) directly, we will reduce (Q1) to the following \emph{batched adjacency query} (Q2). We give an arbitrary order to the groups in ${\cal R}^{(j)}_{\act}$, denoted by ${\cal R}^{(j)}_{\act} = \{R^{(j)}_{1}, R^{(j)}_{2},..., R^{(j)}_{\bar{k}^{(j)}}\}$.
\begin{itemize}
\item[(Q2)] Given a group $R^{(j)}_{k}\in {\cal R}^{(j)}_{\act}$ and a batch of consecutive group $R^{(j)}_{\ell},R^{(j)}_{\ell+1},...,R^{(j)}_{r}\in{\cal R}^{(j)}_{\act}$ s.t. $k\notin[\ell,r]$, decide if there exists $R^{(j)}_{k'}$ s.t. $k'\in[\ell, r]$ and $R^{(j)}_{k'}$ is adjacent to $R^{(j)}_{k}$.
\end{itemize}

\begin{lemma}
At phase $j$, one adjacency query can be reduced to $O(\log \bar{k}^{(j)})$ batched adjacency queries.
\label{lemma:ReduceToBatch}
\end{lemma}

\begin{proof}
Consider an adjacency query for some $R^{(j)}_{k}\in {\cal R}^{(j)}_{\act}$. We can either find some $R^{(j)}_{k'}\in {\cal R}^{(j)}_{\act}$ s.t. $k+1\leq k'\leq \bar{k}^{(j)}$ and $R^{(j)}_{k'}$ is adjacent to $R^{(j)}_{k}$ or claim there is no such $R^{(j)}_{k'}$ in the following way: first fix $\ell = k+1$, and then perform a binary search on $r$ in range $[k+1, \bar{k}^{(j)}]$, in which each binary search step is guided by a batched adjacency query with parameters $k,\ell,r$. Similarly, we can try to find an adjacent group $R^{(j)}_{k'}$ to the left of $R^{(j)}_{k}$ by fixing $r = k-1$ and peforming a binary search on $\ell$ in range $[1,k-1]$. The total number of calls to (Q2) is obviously $O(\log\bar{k}^{(j)})$ in these two binary searches.
\end{proof}

To answer batched adjacency queries in each phase, we will first introduce some \emph{additional structures}, and then use them to design the algorithm answering (Q2), which is formalized in \Cref{lemma:BatchedAdjacencyQuery}. 

\begin{lemma}
There is a deterministic algorithm that computes some additional structures in $O(p^{2}d^{2}\Delta^{2}\log n)$ time to support any batched adjacency query in $O(pd)$ time.
\label{lemma:BatchedAdjacencyQuery}
\end{lemma}

We are now ready to analyse the running time of the \Boruvka's algorithm, which completes the proof of \Cref{lemma:Boruvka}. At each phase $j$, the number of adjacency queries is at most $\bar{k}^{(j)}$ (one for each active group in ${\cal R}^{(j)}$), so the number of batched adjacency queries is $O(\bar{k}^{(j)}\log\bar{k}^{(j)})$ by \Cref{lemma:ReduceToBatch}. Thus the total number of batched adjacency queries is $\sum_{j\geq 1}O(\bar{k}^{(j)}\log\bar{k}^{(j)}) = O(|{\cal I}|\log|{\cal I}|)$ by \Cref{prop:BoruvkaPhasesNumber}. By \Cref{lemma:BatchedAdjacencyQuery}, the total running time of step 1 is $O(p^{2}d^{2}\Delta^{2}\log n) + O(pd|{\cal I}|\log|{\cal I}|) = O(p^{2}d^{2}\Delta^{2}\log n)$. The total running time of the \Boruvka's algorithm is asymptotically the same because step 2 takes little time.

In what follows, we prove \Cref{lemma:BatchedAdjacencyQuery}.

\paragraph{The Additional Structures.} 

We start with introducing some notations. For a group $R\subseteq {\cal I}$, we use $V(R) = \bigcup_{I\in R} I$ denote its vertex set. For two disjoint groups $R_{1}, R_{2}\subseteq {\cal I}$ and a (multi) set $E$ of undirected edges, let $\delta_{E}(R_{1}, R_{2})$ denote the number of edges in $E$ with one endpoint in $V(R_{1})$ and the other one in $V(R_{2})$. Also, recall that we gave an order to groups in ${\cal R}^{(j)}_{\act}$, denoted by ${\cal R}^{(j)}_{\act} = \{R^{(j)}_{1},...,R^{(j)}_{\bar{k}^{(j)}}\}$.

For each phase $j$, we will construct the following data structures.
\begin{itemize}
\item First, we construct a two-dimensional $(\bar{k}^{(j)}\times \bar{k}^{(j)})$-array $\CountAll^{(j)}$, where for each $1\leq x,y\leq \bar{k}^{(j)}$, the entry $\CountAll^{(j)}(x,y) = \delta_{E(\hat{G})}(R^{(j)}_{x}, R^{(j)}_{y})$.
Furthermore, we store the 2D-prefix sum of $\CountAll^{(j)}$. 
\item For each affected component $\gamma$, we prepare a one-dimensional array $\CountA^{(j)}_{\gamma}$ with length $\bar{k}^{(j)}$, where for each $1\leq x\leq \bar{k}^{(j)}$, the entry $\CountA^{(j)}_{\gamma}(x) = |A_{\gamma}\cap V(R^{(j)}_{x})|$. Similarly, we construct an one-dimensional array $\CountB^{(j)}_{\gamma}$ with length $\bar{k}^{(j)}$ in which the entry $\CountB^{(j)}_{\gamma}(x) = |B_{\gamma}\cap V(R^{(j)}_{x})|$. Furthermore, we store the prefix sum of $\CountA^{(j)}_{\gamma}$ and $\CountB^{(j)}_{\gamma}$. 
\end{itemize}

\begin{lemma}
The total construction time of arrays $\CountAll^{(j)}$, $\CountA^{(j)}_{\gamma}$ and $\CountB^{(j)}_{\gamma}$ summing over all phases $j$ and all affected components $\gamma$ is $O(p^{2}d^{2}\Delta^{2}\log d)$.
\label{lemma:AdditionalStructuresTime}
\end{lemma}
\begin{proof}
We first initialize $\CountAll^{(1)}, \CountA^{(1)}_{\gamma}, \CountB^{(1)}_{\gamma}$ for phase 1. For each entry $\CountAll^{(1)}(x,y)$ of $\CountAll^{(1)}$, note that $R^{(1)}_{x}$ and $R^{(1)}_{y}$ are both singleton groups. Let $I_{x}$ and $I_{y}$ be the intervals in $R^{(1)}_{x}$ and $R^{(1)}_{y}$. Then $\CountAll^{(1)}(x,y)$ is exactly the number of $E(\hat{G})$-edges that connecting $I_{x}$ and $I_{y}$, which can be answered by querying $\Table$ in $O(\log n)$ time by \Cref{lemma:2DRangeCounting} because $I_{x}$ and $I_{y}$ are intervals on the global order $\pi$. For an entry $\CountA^{(1)}_{\gamma}(x)$ of $\CountA^{(1)}_{\gamma}$, let $I_{x}$ be the single interval in $R^{(1)}_{x}$, and we can easily compute $|A_{\gamma}\cap I_{x}|$ by binary search in $O(\log n)$ time because $I_{x}$ is an interval on $\pi$ and $A_{\gamma}$ is ordered consistently with $\pi$. Similarly, we can compute the array $\CountB^{(1)}_{\gamma}$. The construction time of additional structures at phase 1 is $O((\bar{k}^{(1)})^{2} + |{\cal C}_{\aff}|\cdot\bar{k}^{(1)})\log n)$.

For each phase $j\geq 2$, we will compute $\CountAll^{(j)}, \CountA^{(j)}_{\gamma}, \CountB^{(j)}_{\gamma}$ based on the arrays of phase $j-1$. For an entry $\CountAll^{(j)}(x,y)$ of $\CountAll^{(j)}$, recall that $R^{(j)}_{x}$ is the union of several groups $R^{(j-1)}_{x_{1}},R^{(j-1)}_{x_{2}},...$ inside ${\cal R}^{(j-1)}_{\act}$, and $R^{(j)}_{y} = R^{(j-1)}_{y_{1}}\cup R^{(j-1)}_{y_{2}}\cup ...$. Furthermore, $x_{1},x_{2},...,y_{1},y_{2},...$ are distinct indexes in $[1,\bar{k}^{(j-1)}]$. Therefore,
\[
\CountAll^{(j)}(x,y) = \sum_{x'=x_{1},x_{2},...}\sum_{y'=y_{1},y_{2},..}\CountAll^{(j-1)}(x',y').
\]
We can compute $\CountA^{(j)}_{\gamma}$ and $\CountB^{(j)}_{\gamma}$ in a similar way. The construction time of additional structures at phase $j$ is proportional to the the total size of additional structures at phase $j-1$, i.e. $O((\bar{k}^{(j-1)})^{2} + |{\cal C}_{\aff}|\cdot\bar{k}^{(j-1)})$.

The overall construction time is 
\[
O((\bar{k}^{(1)})^{2} + |{\cal C}_{\aff}|\cdot\bar{k}^{(1)})\log n) + \sum_{j\geq 2}O((\bar{k}^{(j-1)})^{2} + |{\cal C}_{\aff}|\cdot\bar{k}^{(j-1)}) = O(p^{2}d^{2}\Delta^{2}\log n),
\]
because $\bar{k}^{(1)} = |{\cal I}| = O(pd\Delta)$, $|{\cal C}_{\aff}| = O(pd)$ and $\bar{k}^{(j)}\leq \bar{k}^{(j-1)}/2$ for each phase $j$.
\end{proof}

\paragraph{Answering Batched Adjacency Queries.} Consider a batched adjacency query at phase $j$ with parameters $k,\ell,r$. It is equivalent to decide whether the number of $G^{\star}$-edges connecting $R^{(j)}_{k}$ and some $R^{(j)}_{k'}$ where $k'\in[\ell,r]$ is greater than zero or not. Namely, it suffices to decide whether
\begin{equation}
\sum_{\ell\leq k'\leq r} \delta_{E(G^{\star})}(R^{(j)}_{k},R^{(j)}_{k'})>0.
\label{eq:BatchedAdjacencyQuery}
\end{equation}

\begin{lemma}
For any two disjoint groups $R_{1},R_{2}\subseteq {\cal I}$, 
\[
\delta_{E(G^{\star})}(R_{1}, R_{2}) = \delta_{E(\hat{G})}(R_{1}, R_{2}) - \sum_{\gamma\in{\cal C}_{\aff}}\delta_{\hat{E}_{\gamma}}(R_{1}, R_{2}).
\]

\label{lemma:6.8}
\end{lemma}
\begin{proof}
First the RHS is equal to $\delta_{E(\hat{G}) - \sum_{\gamma\in{\cal C}_{\aff}}\hat{E}_{\gamma}}(R_{1}, R_{2})$ because $\hat{E}_{\gamma}$ of all $\gamma\in{\cal C}_{\aff}$ are disjoint subsets of $E(\hat{G})$ (note that $E(\hat{G})$ is defined to be a multiset). 

Recall that $G^{\star} = \hat{G}[Q^{\star}] - \sum_{\gamma\in{\cal C}_{\aff}}\hat{E}_{\gamma}$.
The LHS is at most the RHS because $E(G^{\star})\subseteq E(\hat{G})-\sum_{\gamma\in{\cal C}_{\aff}}\hat{E}_{\gamma}$. On the other direction, each edge in $E(\hat{G})$ connecting $V(R_{1})$ and $V(R_{2})$ is inside $\hat{G}[Q^{\star}]$ since $V(R_{1}),V(R_{2})\subseteq Q^{\star}$, so the RHS is at most the LHS.
\end{proof}

\begin{lemma}
For each $\gamma\in{\cal C}_{\aff}$ and two disjoint groups $R_{1},R_{2}\subseteq {\cal I}$, 
\[
\delta_{\hat{E}_{\gamma}}(R_{1},R_{2}) = |A_{\gamma}\cap V(R_{1})|\cdot |A_{\gamma}\cap V(R_{2})| - |(A_{\gamma}\setminus B_{\gamma})\cap V(R_{1})|\cdot|(A_{\gamma}\setminus B_{\gamma})\cap V(R_{2})|.\]
\label{lemma:6.9}
\end{lemma}
\begin{proof}
Recall that $\hat{E}_{\gamma}$ is the union of a clique on $B_{\gamma}$ and a biclique between $A_{\gamma}\setminus B_{\gamma}$ and $B_{\gamma}$. In other words, $\hat{E}_{\gamma}$ is a clique on $A_{\gamma}$ with the clique on $A_{\gamma}\setminus B_{\gamma}$ removed. Because $V(R_{1})$ and $V(R_{2})$ are disjoint, the equation follows.
\end{proof}

Using \Cref{lemma:6.8} and \Cref{lemma:6.9}, we can rewrite the LHS of inequality \ref{eq:BatchedAdjacencyQuery} as follows.
\begin{align*}
\sum_{\ell\leq k'\leq r} \delta_{E(G^{\star})}(R^{(j)}_{k}, R^{(j)}_{k'}) &= \sum_{\ell\leq k'\leq r} \delta_{E(\hat{G})}(R^{(j)}_{k}, R^{(j)}_{k'}) - \sum_{\ell\leq k'\leq r}\sum_{\gamma\in{\cal C}_{\aff}}\delta_{\hat{E}_{\gamma}}(R^{(j)}_{k}, R^{(j)}_{k'})\\
&= \sum_{\ell\leq k'\leq r} \delta_{E(\hat{G})}(R^{(j)}_{k}, R^{(j)}_{k'})- \sum_{\gamma\in{\cal C}_{\aff}}\sum_{\ell\leq k'\leq r}|A_{\gamma}\cap V(R^{(j)}_{k})|\cdot|A_{\gamma}\cap V(R^{(j)}_{k'})| \\
&+\sum_{\gamma\in{\cal C}_{\aff}}\sum_{\ell\leq k'\leq r} |(A_{\gamma}\setminus B_{\gamma})\cap V(R^{(j)}_{k})|\cdot|(A_{\gamma}\setminus B_{\gamma})\cap V(R^{(j)}_{k'})|.\\
\end{align*}

Combining the definition of the additional structures, we further have
\begin{align*}
\sum_{\ell\leq k'\leq r} \delta_{E(G^{\star})}(R^{(j)}_{k}, R^{(j)}_{k'})&=\sum_{\ell\leq k'\leq r}\CountAll^{(j)}(k,k') - \sum_{\gamma\in{\cal C}_{\aff}}\left(\CountA^{(j)}_{\gamma}(k)\cdot\sum_{\ell\leq k'\leq r}\CountA^{(j)}_{\gamma}(k')\right)\\
&+\sum_{\gamma\in{\cal C}_{\aff}}\left((\CountA^{(j)}_{\gamma}(k)-\CountB^{(j)}_{\gamma}(k))\cdot \sum_{\ell\leq k\leq r}(\CountA^{(j)}_{\gamma}(k')-\CountB^{(j)}_{\gamma}(k'))\right).
\end{align*}
Because we have stored the prefix sum of the arrays $\CountAll^{(j)}, \CountA^{(j)}_{\gamma}, \CountB^{(j)}_{\gamma}$, computing the value of the above expression takes $O(|{\cal C}_{\aff}|) = O(pd)$ time.

\subsection{The Query Algorithm}
\label{sect:QueryAlgorithm}

Let $u,v\in V_{\new}$ be a given query. The query strategy is to find some $u^{\star},v^{\star}\in Q^{\star}$ s.t. $u$ and $u^{\star}$ (resp. $v$ and $v^{\star}$) are connected in $G[V_{\new}]$. Intuitively, the connectivity of $u^{\star}$ and $v^{\star}$ on $G^{\star}$ can be answered by \Cref{lemma:Boruvka}, which implies the connectivity of $u$ and $v$ in $G[V_{\new}]$ by \Cref{lemma:ConnEq}.

We take $u$ as an example and discuss how to find a valid $u^{\star}$. If $u\in Q^{\star}$, then we can trivially let $u^{\star} = u$. Otherwise, $u$ must belong to $U(\gamma_{u})$ of some unaffected component $\gamma_{u}$, because $u\in V_{\new}\setminus Q^{\star}$ and
\[
V_{\new}\setminus Q^{\star} = (V_{\on}\setminus D_{\on})\setminus \bigcup_{\tau\in {\cal T}_{\aff}}U(\tau)\subseteq (V_{\on}\setminus D_{\on})\setminus \bigcup_{\gamma\in {\cal C}_{\aff}}U(\gamma) = \bigcup_{\text{unaffected }\gamma} U(\gamma)\setminus D_{\on}.
\]
Let $\hat{\gamma}_{u}$ be the maximal unaffected component containing $u$. 

We will decide whether $A_{\hat{\gamma}_{u}}\cap V_{\new} = \emptyset$ or not as follows. Note that $A_{\hat{\gamma}_{u}}\cap V_{\new} = (A_{\hat{\gamma}_{u},\off}\cap D_{\off})\cup(A_{\hat{\gamma}_{u},\on}\setminus D_{\on})$. For each $v\in D_{\off}$, we have an indicator that indicates whether $v\in A_{\hat{\gamma}_{u},\off}$, so deciding whether $A_{\hat{\gamma}_{u},\off}\cap D_{\off}=\emptyset$ takes $O(|D_{\off}|)$ time. To decide whether $A_{\hat{\gamma}_{u},\on}\setminus D_{\on} = \emptyset$, we can simply scan the first $\min\{|D_{\on}|+1,|A_{\hat{\gamma}_{u},\on}|\}$ vertices in $A_{\hat{\gamma}_{u},\on}$, which takes $O(|D_{\on}|)$ time. Therefore, deciding whether $A_{\hat{\gamma}_{u}}\cap V_{\new} = \emptyset$ takes $O(d)$ time.

\begin{itemize}
\item Suppose $A_{\hat{\gamma}_{u}}\cap V_{\new} \neq \emptyset$. The above algorithm can also find some $u^{\star}\in A_{\hat{\gamma}_{u}}\cap V_{\new}$. The correctness of $u^{\star}$ is as follows. By property (1) of the low degree hierarchy, each vertex in $A_{\hat{\gamma}_{u},\on}$ belongs to $U(\gamma')$ of some ancestor $\gamma'$, which implies $A_{\hat{\gamma}_{u},\on}\subseteq \bigcup_{\gamma\in{\cal C}_{\aff}} U(\gamma)$ because $\hat{\gamma}_{u}$ is a maximal unaffected component. Then $A_{\hat{\gamma}_{u},\on}\cap V_{\new} = (A_{\hat{\gamma}_{u},\off}\cap D_{\off})\cup (A_{\hat{\gamma}_{u},\on}\setminus D_{\on})\subseteq Q^{\star}$. Furthermore, because $\hat{\gamma}_{u}$ is unaffected, $u$ and $u^{\star}$ are connected in $G[V_{\new}]$.

\item Suppose $A_{\hat{\gamma}_{u}}\cap V_{\new} = \emptyset$. If $v$ is also inside $\hat{\gamma}_{u}$, then we answer $u$ and $v$ are connected in $G[V_{\new}]$, otherwise they are disconnected. The query algorithm terminates in this case. 
\end{itemize}

From now, we assume that $u^{\star}$ and $v^{\star}$ are found successfully. Because we have stored the low degree hierarchy explicitly, we can find the trees $\tau_{u},\tau_{v}\in{\cal T}$ s.t. $u^{\star}\in U(\tau_{u^{\star}})$ and $v^{\star}\in U(\tau_{v^{\star}})$ in $O(1)$ time. Note that $\tau_{u^{\star}},\tau_{v^{\star}}\in{\cal T}_{\aff}$ because $u^{\star},v^{\star}\in Q^{\star}$. By \Cref{lemma: Interval}, we can obtain the interval $I_{u^{\star}}\in{\cal I}$ s.t. $u^{\star}$ is connected to $I_{u^{\star}}$ in $\tau_{u^{\star}}\setminus D_{\on}$ in $O(\log d)$ time, which means $u^{\star}$ is connected to $I_{u^{\star}}$ in $G^{\star}$ by \Cref{lemma:ConnEq}. Similarly, we can find such interval $I_{v^{\star}}\in{\cal I}$ for $v^{\star}$. The connectivity of $u^{\star}$ and $v^{\star}$ on $G^{\star}$ is exactly the connectivity of $I_{u^{\star}}$ and $I_{v^{\star}}$ on $G^{\star}$, which can be answered by the output of \Cref{lemma:Boruvka} in $O(1)$ time. By \Cref{lemma:ConnEq} and the properties of $u^{\star}$ and $v^{\star}$, this implies the connectivity of $u$ and $v$ on $G[V_{\new}]$. The total query time is $O(d + \log d) = O(d)$.

\section{Conditional Lower Bounds}
\label{sect:LB}

In this section, we present conditional lower bounds for the subgraph connectivity sensitivity oracles problem in both the decremental and fully dynamic setting. The conditional lower bounds for the decremental setting have been well-studied by a sequence of works \cite{HenzingerKNS15,duan2020connectivity,long2022near} as summarized below.

\begin{theorem}[Theorem 1.2 in \cite{long2022near}]
\label{thm:VertexFailureLower}
Let $\A$ be any vertex-failure connectivity
oracle with $t_{p}$ preprocessing time, $t_{u}$ update time, and
$t_{q}$ query time bound. Assuming popular conjectures, we have the
following:
\begin{enumerate}
\item $\A$ must take $\Omega(\min\{m,nd_{\star}\})$ space. 
\item If $t_{p}=\poly(n)$, then $t_{u}+t_{q}=\Omegahat(d^{2})$.
\item If $t_{p}=\poly(n)$ and $t_{u}=\poly(dn^{o(1)})$, then $t_{q}=\Omegahat(d)$.
\item If $t_{u},t_{q}=\poly(dn^{o(1)})$ and $\A$ is a combinatorial algorithm.
\end{enumerate}
\end{theorem}

Note that the conditional lower bounds in \Cref{thm:VertexFailureLower} also hold for the fully dynamic subgraph connectivity problem in the sensitivity setting, because the vertex-failure connectivity oracles problem is just a special case when $V_{\off}$ is always empty. In what follows, we will discuss some stronger conditional lower bounds on the preprocessing time and the space for the fully dynamic subgraph connectivity problem in the sensitivity setting.

\paragraph{Preprocessing Time.} Regarding the preprocessing time, a previous work \cite{HKP23} shows the following conditional lower bound in terms of $n_{\off}$ and $m$.

\begin{lemma}[Theorem 2 in \cite{HKP23}]
Unless the 3SUM conjectures fails, for any constant $\epsilon > 0$ and $n$-vertices $m$-edges graphs with $n_{\off}$ initial off-vertices, there is no fully dynamic sensitivity oracle for subgraph connectivity with preprocessing time $O((n_{\off} + d)^{1-\epsilon}\cdot m)$ or $O((n_{\off} + d)\cdot m^{1-\epsilon})$, update time and query time of the form $f(d)\cdot n^{o(1)}$.
\label{lemma:PreprocessLB3SUM}
\end{lemma}

We now show another lower bound by reducing boolean matrix multiplication to this problem.

\begin{lemma}
For $n$-vertices $m$-edges graphs where $m=\Theta(n^{2})$, there is no fully dynamic sensitivity oracle for subgraph connectivity with preprocessing time $O(n^{\omega_{\bool}-\epsilon})$ and update and query time of the form $f(d)\cdot n^{o(1)}$, for any constant $\epsilon > 0$, where $\omega_{\bool}$ is the exponent of Boolean matrix multiplication.
\label{lemma:PreprocessLBBMM}
\end{lemma}
\begin{proof}
We assume that $\omega>2$, otherwise this lemma is trivial because at least we need to read the whole graph, which already takes $\Theta(n^{2})$ time.

Assume for the contradiction that there exists such oracle ${\cal A}$ with preprocessing time $O(n^{\omega-\epsilon})$ for some sufficiently small constant $\epsilon >0$. Consider an arbitrary Boolean matrix multiplication instance $(X,Y)$, where $X$ and $Y$ are two $\hat{n}\times\hat{n}$ matrices. We construct a layer graph $H$ as follows. The vertices are $V(H) = \{a_{i}\mid 1\leq i\leq \hat{n}\}\cup\{b_{j}\mid 1\leq j\leq \hat{n}\}\cup\{c_{k}\mid 1\leq k\leq \hat{n}\}$. The edges are $\{\{a_{i},b_{j}\}\mid X_{ij}=1\}\cup\{\{b_{j},c_{k}\}\mid Y_{jk}=1\}$. We let $\hat{H}$ be the graph by adding an isolated dummy clique with $\hat{n}$ vertices into $H$. We define $\hat{H}$ because, after adding a dummy clique, now $|V(\hat{H})| = 4\hat{n}$ and $|E(\hat{H})| = \Theta(\hat{n}^{2}) = \Theta(|V(\hat{H})|^{2})$, which satisfies the requirement of applying ${\cal A}$ on $\hat{H}$.

We construct the oracle ${\cal A}$ on $\hat{H}$ with initial on-vertices $V_{\on}(\hat{H}) = \{b_{j}\mid 1\leq j\leq\hat{n}\}$ and initial off-vertices $V_{\off}(\hat{H}) = V(\hat{H})\setminus V_{\on}(\hat{H})$. To compute the matrix $Z = X\cdot Y$, we consider each entry $z_{ik}$ separately. We switch two vertices $D = \{a_{i},c_{k}\}$. After the update $D$, the new on-vertices are $V_{\new} = \{a_{i},c_{k}\}\cup\{b_{j}\mid 1\leq j\leq \hat{n}\}$. Observe that $z_{ik}=1$ if and only if $a_{i}$ and $c_{k}$ are connected in $\hat{H}[V_{\new}]$. Therefore, deciding $z_{ik}$ just needs one update operation and one query operation on ${\cal A}$, which takes $f(d)\cdot |V(\hat{H})|^{o(1)} = \hat{n}^{o(1)}$ time because $d = |D| = 2$. The total time to compute $Z$ is $O(|V(\hat{H})|^{\omega-\epsilon}) + \hat{n}^{2}\cdot \hat{n}^{o(1)} = O(\hat{n}^{\omega-\epsilon})$, contradicting the definition of $\omega$. 
\end{proof}

\begin{remark}
\label{remark:}
We emphasize that \Cref{lemma:PreprocessLB3SUM} and \Cref{lemma:PreprocessLBBMM} are not contradictory, because the lower bound in \Cref{lemma:PreprocessLB3SUM} is obtained when the input graphs are sparse, but it does not hold if we restrict the input graphs to be dense. Concretely, by checking the proof of \Cref{lemma:PreprocessLB3SUM} in \cite{HKP23}, we have the following refined version \Cref{lemma:RefinedPreprocessLB3SUM}. \Cref{lemma:RefinedPreprocessLB3SUM} implies \Cref{lemma:PreprocessLB3SUM} because to obtain the statement in \Cref{lemma:PreprocessLB3SUM} for an arbitrary constant $\epsilon > 0$, we just need to consider the statement in \Cref{lemma:RefinedPreprocessLB3SUM} for $\epsilon' = \min\{\epsilon/10,1/2\}$.

\begin{lemma}[Theorem 2 in \cite{HKP23}, Refined]
Unless the 3SUM conjectures fails, for any constant $0<\epsilon\leq 1/2$ and $n$-vertices $m$-edges graphs with $n_{\off}$ initial off-vertices, where $n_{\off} = \Theta(n)$ and $m=\wtilde{\Theta}(n^{1+\epsilon})$, there is no fully dynamic sensitivity oracle for subgraph connectivity with preprocessing time $O(n^{2-\epsilon})$, update time and query time of the form $f(d)\cdot n^{o(1)}$.
\label{lemma:RefinedPreprocessLB3SUM}
\end{lemma}
\end{remark}

Combining \Cref{lemma:RefinedPreprocessLB3SUM} and \Cref{lemma:PreprocessLBBMM}, we can see that our preprocessing time $\wtilde{O}(\min\{m(n_{\off} + d),n^{\omega}\})$ is somewhat tight. At least it is tight up to subpolynomial factors at the extreme point where the input graphs has $n_{\off} = \Theta(n)$ and number of edges roughly linear to $n$. For the another extreme point where the input graphs has $m=\Theta(n^{2})$, the upper bound $\wtilde{O}(n^{\omega})$ and the lower bound $\what{\Omega}(n^{\omega_{\bool}})$ are kind of matched (To our best knowledge, currently there is no clear separation between $\omega$ and $\omega_{\bool}$). 

\begin{conjecture}[No truly subcubic combinatorial BMM, \cite{AW14}]
In the Word RAM model with words of $O(\log n)$ bits, any combinatorial algorithm requires $n^{3-o(1)}$ time in expectation to compute the Boolean product of two $n\times n$ matrices.
\label{conj:BMM}
\end{conjecture}

\begin{lemma}
Assuming \Cref{conj:BMM}, for any constant 
$0\leq \delta\leq 1$ and $n$-vertices $m$-edges graphs with $n_{\off}$ initial off-vertices, where $m=\Theta(n^{1+\delta})$ and $n_{\off} = \Theta(n^{\delta})$, there is no combinatorial fully dynamic sensitivity oracle for subgraph connectivity with preprocessing time $O(n^{1+2\delta-\epsilon})$ and update and query time of the form $f(d)\cdot n^{o(1)}$, for any constant $\epsilon > 0$.
\label{lemma:PreprocessCombLBBMM}
\end{lemma}

\Cref{lemma:PreprocessCombLBBMM} can be proved in the way similar to \Cref{lemma:PreprocessLBBMM}, by considering multiplying an $n_{\off}\times n$ Boolean matrix and an $n\times n_{\off}$ Boolean matrix.

\begin{remark}
Additionally, our first upper bound $\hat{O}(m(n_{\off} + d))$ is obtained from an combinatorial algorithm. The term $n_{\off}m$ cannot avoid for input graphs with an arbitrary fixed density by \Cref{lemma:PreprocessCombLBBMM}. Another term $dm$ is also hard to avoid because even the vertex-failure connectivity oracles problem has a combinatorial lower bound $\hat{\Omega}(md)$ on preprocessing time by item 4 in \Cref{thm:VertexFailureLower}. 
\label{remark:CombUP}
\end{remark}

\paragraph{Space.} For the space complexity, we show that even when considering sparse input graphs, the oracle still need $\wtilde{\Omega}(n^{2})$ space. This lower bound is conditioning on the strong set disjointness conjecture \cite{GKLP17}. Basically, in a set disjointness problem, we need to preprocess a family $F$ of sets with total size $N_{F} = \sum_{S\in F}|S|$, all from universe $U$, so that given two query sets $S,S'\in F$, we can determine if $S\cap S'=\emptyset$.

\begin{conjecture}[Strong Set Disjointness conjecture \cite{GKLP17}]
Any data structure for the set disjointness problem that answers queries in $T$ time must use $\wtilde{\Omega}(N^{2}_{F}/T^{2})$ space.
\label{conj:SetDisjointness}
\end{conjecture}

In \Cref{lemma:SpaceLB}, we only need to assume \Cref{conj:SetDisjointness} holds for some large enough $T = N_{F}^{o(1)}$.

\begin{lemma}
Assuming \Cref{conj:SetDisjointness} for some large enough $T = N^{o(1)}_{F}$, for $n$-vertices $m$-edges graphs where $m = \Theta(n)$, there is no fully dynamic sensitivity oracle for subgraph connectivity with space $O(n^{2-\epsilon})$, update time and query time of the form $f(d)\cdot n^{o(1)}$ for any constant $\epsilon>0$.
\label{lemma:SpaceLB}
\end{lemma}
\begin{proof}
Assume for the contradiction that, for $n$-vertices and $m$-edges graphs where $m=\Theta(n)$, there is an oracle ${\cal A}$  with space $O(n^{2-\epsilon})$ for some constant $\epsilon>0$. Consider an arbitrary set disjointness instance $(F,U)$. Note that trivially we have $|F|,|U|\leq N_{F}$. We construct a graph $H$ as follows. The vertex set of $H$ is $V(H) = \{v_{S}\mid S\in F\}\cup\{v_{u}\mid u\in U\}\cup V'$, where $V'$ contains $2N_{F} - |F|-|U|$ many dummy vertices. The edges are $E(H) = \{\{v_{S},v_{u}\}\mid S\in F,u\in U\text{ s.t. }u\in S\}$. In other words, each set in $F$ and element in $U$ corresponds to a distinct vertex in $H$, but there may be some dummy vertices corresponding to nothing. Edges represent those pairs of $u,S$ s.t. $u\in S$. We point out that adding the dummy vertices $V'$ is just for ensuring the sparsity of $H$, i.e. $|E(H)| = \Theta(|V(H)|)$.

We now show that the oracle ${\cal A}$ can be used to answer the set disjointness queries. We let the initial off-vertices be $V_{\off}(H) = \{v_{S}\mid S\in F\}$ and the initial on-vertices be $V_{\on}(H) = \{v_{u}\mid u\in U\}\cup V'$. When a query $S,S'$ comes, we let $D = \{v_{S},v_{S'}\}$ be the vertices we will switch. The new set of on-vertices after the update is then $V_{\new} = \{v_{S},v_{S'}\}\cup\{v_{u}\mid u\in U\}\cup V'$. Observe that $S\cap S'\neq\emptyset$ if and only if $v_{S}$ and $v_{S'}$ are connected in $H[V_{\new}]$. 

Answering this set disjointness query needs one update operation and one query operation to ${\cal A}$, which takes $f(d)\cdot |V(H)|^{o(1)} = N_{F}^{o(1)}$ time because $d = |D| = 2$ and $|V(H)| = 2N_{F}$. The data structure ${\cal A}$ takes space $O(|V(H)|^{2-\epsilon}) = O(N_{F}^{2-\epsilon})$, contradicting \Cref{conj:SetDisjointness}.

\end{proof}

We can easily extend \Cref{lemma:SpaceLB} to make it work for input graph families with arbitrary fixed densities, as shown in \Cref{coro:SpaceLBAllDensity}.

\begin{corollary}
Assuming \Cref{conj:SetDisjointness} for some large enough $T = N^{o(1)}_{F}$, for any constants $1\leq \delta\leq 2,\ \epsilon>0$ and $n$-vertices $m$-edges graphs where $m = \Theta(n^{\delta})$, there is no fully dynamic sensitivity oracle for subgraph connectivity with space $O(n^{2-\epsilon})$, update time and query time of the form $f(d)\cdot n^{o(1)}$ for any constant $\epsilon>0$.
\label{coro:SpaceLBAllDensity}
\end{corollary}
\begin{proof}
Fix a density parameter $\delta$. Assume for contradiction that such oracles ${\cal A}$ exists for some constant $\epsilon > 0$. Consider an arbitrary input instance $H$ of \Cref{lemma:SpaceLB}, where $H$ is a graph with density $|E(H)|=\Theta(|V(H)|)$. We extend $H$ to an input instance $\hat{H}$ of \Cref{coro:SpaceLBAllDensity} by adding an isolated dummy graph $H'$ to $H$ with $|V(H')| = |V(H)|$ and $|E(H')| = \Theta(|V(H')|^{\delta})$. Adding this isolated dummy graph $H'$ is just to ensure that $\hat{H}$ satisfies the sparsify requirement of \Cref{coro:SpaceLBAllDensity}, i.e. $|E(\hat{H})| = \Theta(|V(\hat{H})|^{\delta})$. Obviously, the oracle ${\cal A}$ for $\hat{H}$ is also a valid oracle for $H$ with the same space, update time and query time. Hence, we obtain an oracle for $H$ with space $O(|V(\hat{H})|^{2-\epsilon}) = O(|V(H)|^{2-\epsilon})$ because $|V(H)| = \Theta(|V(\hat{H})|)$, contradicting \Cref{lemma:SpaceLB}.
\end{proof}

\Cref{coro:SpaceLBAllDensity} shows that the space complexity of our oracle, i.e. $\wtilde{O}(\min\{m(n_{\off}+d),n^{2}\})$ is tight up to subpolynomial factors, for input graph families with general densities.

\section*{Acknowledgements}

We thank Thatchaphol Saranurak for helpful discussions.

\clearpage

\appendix

\section{Omitted Proofs}
\label{sect:OmittedProofs}

\subsection{Proof of \Cref{lemma:MatchingPlayer}}
\label{sect:MatchingPlayerProof}
\begin{lemma}[Deterministic Approximate Vertex-Capacitated Max Flow, Lemma B.2 in \cite{long2022near}] Let $G=(V(G),E(G),c)$ be an $n$-vertex $m$-edge undirected graph with vertex capacity function $c:V(G)\to [1,C]\cup\{\infty\}$. Let $s,t\in V(G)$ be source and sink vertices such that $s$ and $t$ are not directly connected by an edge and $c(s),c(t)=\infty$, and they are the only vertices with infinite capacity. Then, given a parameter $\epsilon$ s.t. $1/\polylog(n)< \epsilon < 1$, there is a deterministic algorithm that computes,
\begin{itemize}
    \item a feasible flow $f$ that sends a $(1-\epsilon)$-fraction of the max flow value from $s$ to $t$, and
    \item a cut $S\subseteq V(G)\setminus\{s,t\}$ whose deletion will disconnect $s$ from $t$, with total capacities no more than $(1+\epsilon)$ times the min cut value.
\end{itemize}
The algorithm runs in $m^{1+o(1)}\cdot\polylog C+n^{o(1)}\cdot C_{\ssum}$ time, where $C$ is the maximum capacity and $C_{\ssum}=\sum_{v\in V(G)\setminus\{s,t\}}c(v)$ denotes the total capacities of finite capacitated vertices.
\label{lemma:DetVertexFlow}
\end{lemma}

\begin{proof}[Proof of \Cref{lemma:MatchingPlayer}]
Without loss of generality, we assume $|A|\leq |B|$. We first construct an auxiliary vertex-capacitated undirected graph $G' = (V(G'),E(G'),c)$ as follows. The vertex set is $V(G') = \{s,t\}\cup\{u_{A}\mid u\in A\}\cup \{v_{B}\mid v\in B\}\cup V(G)$, where (i) $s$ and $t$ are the source and sink vertex with capacity $c(s) = c(t) = \infty$; (2) $\{u_{A}\mid u\in A\}$ and $\{v_{B}\mid v\in B\}$ are artificial vertices with capacity $1$, which one-one corresponds to vertices in $A$ and $B$; (3) $V(G)$ are original graph vertices and each of them has capacity set to be $\lceil 1/\phi\rceil$. The edge set $E(G') = \{\{s,u_{A}\},\{u_{A},u\}\mid u\in A\}\cup E(G)\cup \{\{v,v_{B}\},\{v_{B},t\}\mid v\in B\}$. Note that there is no capacity constraint on edges.

We run the approximate max flow algorithm from \Cref{lemma:DetVertexFlow} on $G'$ with parameter $\epsilon = 1/10$. Let the output be a feasible $s$-$t$ flow $f$ on $G'$ and a vertex cut $(L',S',R')$ of $G'$ with $s\in L'$ and $t\in R'$. Precisely, \Cref{lemma:DetVertexFlow} only gives $S'$, and we can let $L'\subseteq V(G')$ be the vertices connected to $s$ in $G'\setminus S'$ and $R' = V(G')\setminus(L'\cup S')$. The algorithm will guarantee that $\val(f)\geq \frac{1-\epsilon}{1+\epsilon}\cdot c(S')\geq c(S')/2$. We then consider the following two cases.

\

\noindent\textbf{Case (1): $\val(f)< |A|/3$.} In this case, we have $c(S')\leq 2\cdot\val(f)\leq 2|A|/3$. We define the vertex cut $(L,S,R)$ on $G$ by letting $L = L'\cap V(G)$, $S = S'\cap V(G)$ and $R = R'\cap V(G)$. Note that $(L,S,R)$ is a vertex cut of $G$ follows straightforwardly that $(L',S',R')$ is a vertex cut on $G'$. 

Now we show $|L\cap U|\geq |A|/3$. First, there are at most $c(S')$ many terminals $u\in A$ s.t. $u_{A}\in S'$ or $u\in S'$ because both $c(u_{A}),c(u)\geq 1$. In other words, there are at least $|A|-c(S')$ many terminals $u\in S'$ s.t. $u_{A}\notin S'$ and $u\notin S'$, which means $u\in L'$. This concludes $|L\cap U|\geq |A|-c(S')\geq |A|/3$. A similar argument will lead to $|R\cap U|\geq |B|-c(S')\geq |A|/3$, because $|B|\geq |A|$. Hence, the cut $(L,S,R)$ is exactly what we desire (if $|L\cup U|>|R\cup U|$, we can just swap $L$ and $R$).

\

\noindent\textbf{Case (2): $\val(f)\geq |A|/3$.} We first round $f$ into an integral feasible flow $f_{\it}$ with value $\val(f_{\it})\geq \lceil |A|/3\rceil$ using the flow rounding algorithm in \cite{kang2015flow}, which takes $\tilde{O}(m)$ time. Then we decompose the integral flow $f_{\it}$ into flow paths implicitly (that is, for each flow path, we only computes its endpoints $u_{A}$ and $v_{B}$, ignoring the trivial endpoints $s$ and $t$), which can be done in $\tilde{O}(m)$ time using dynamic trees \cite{sleator1981data}. This leads to a matching $M$ between $A$ and $B$ s.t. $|M|=\val(f_{\it})\geq |A|/3$ and there is an embedding $\Pi_{M\to G}$ of $M$ into $G$ with vertex congestion $\lceil 1/\phi\rceil$ (because in $G'$, each original graph vertex has capacity $\lceil 1/\phi\rceil$). This matching $M$ is exactly what we want. The edges $E(\Pi_{M\to G})$ used by the embedding $\Pi_{M\to G}$ are exactly the original graph edges $e$ with $f_{\it}(e)>0$.

\

The running time is dominated by the time taken by the approximate max flow algorithm, which is $m^{1+o(1)}$ by \Cref{lemma:DetVertexFlow}.
\end{proof}

\subsection{Proof of \Cref{thm:LowDegreeHierarchy}}
\label{sect:CorrectnessLowDegreeHierarchy}

\begin{proof}
The components ${\cal C}$ forms a laminar set (i.e. it satisfies properties (1) and (2)) because $V(G) = U'_{1}\supseteq U'_{2}\supseteq ... \supseteq U'_{p}$ and the level-$i$ components are connected components of $G\setminus U'_{i+1}$. In particular, $V(G)$ is the unique component at the top level $p$ because $U'_{p+1} = \emptyset$. Furthermore, for each level $i$, the algorithm constructs $U_{i}$ as the terminal set at level $i$. For each $\gamma\in {\cal C}_{i}$, if we let $U(\gamma) = U_{i}\cap V(\gamma)$, then it will be consistent with the definition in (3). 

Next, we show the properties of Stenier trees. Property (4) is because the Steiner trees outputted by \Cref{lemma:SFDecomp} are vertex-disjoint. Property (5) is just a definition. It remains to show property (6). Consider a component $\gamma\in {\cal C}_{i}$ with $U(\gamma)\neq \emptyset$. By the algorithm, $\gamma$ is a connected component of $G\setminus U'_{i+1}$. It must be inside some connected component $Y$ of $G\setminus X_{i+1}$ because $X_{i+1}\subseteq U'_{i+1}$. From $U(\gamma) = U_{i}\cap V(\gamma) = (X_{i}\setminus (X_{i+1}\cup ...\cup X_{p}))\cap \gamma\neq\emptyset$, we have $X_{i}\cap V(Y)\neq\emptyset$. Therefore, by \Cref{lemma:SFDecomp}, there is a Steiner tree $\tau_{Y}\in{\cal T}_{i}$ spanning $(X_{i}\setminus X_{i+1})\cap V(Y)$ on $Y$ with maximum degree $O(\log^{2}|U|/\epsilon)$. In (5), we define $U(\tau_{Y}) = U_{i}\cap V(\tau_{Y})$. Because $U_{i} = X_{i}\setminus(X_{i+1}\cup ...\cup X_{p})$, $(X_{i}\setminus X_{i+1})\cap V(Y)\subseteq V(\tau_{Y})\subseteq V(Y)$, we have $U(\tau_{Y}) = U_{i}\cap V(Y)$, which means $U(\gamma) = U_{i}\cap V(\gamma) \subseteq U(\tau_{Y})$.
\end{proof}

\newcommand{\etalchar}[1]{$^{#1}$}

\end{document}